\documentclass[10pt, onecolumn]{IEEEtran}
\usepackage{graphicx}

\usepackage{amsmath}
\usepackage{amssymb}
\usepackage{amsfonts}
\usepackage{epsfig}

\usepackage{amsmath,mathtools}    
\usepackage{graphicx}   
\usepackage{color}      
\usepackage{hyperref}
\usepackage{comment}
\usepackage{subfigure}
\usepackage{times}
\usepackage{hyphenat}
\usepackage{epsfig}
\usepackage{latexsym}
\usepackage{bbold,bbm}
\usepackage{setspace}
\usepackage{amssymb,amsmath}
\usepackage{mathrsfs}
\usepackage{amsgen,amsfonts,amsbsy,amsthm}
\usepackage{algorithmic,algorithm}
\usepackage{array}
\usepackage{booktabs}
\usepackage{amsfonts}
\usepackage{amssymb}
\usepackage{multirow}
\usepackage{longtable}
\usepackage{epstopdf}
\usepackage{wrapfig}
\usepackage[font=small]{caption}

\newtheorem{example}{Example}
\newtheorem{defn}{Definition}
\newtheorem{thm}{Theorem}

\newtheorem{cor}[thm]{Corollary}
\newtheorem{note}{Remark}

\newtheorem{const}{Construction}

\newcommand{\bit}{\begin{itemize}}
\newcommand{\eit}{\end{itemize}}
\newcommand{\bcor}{\begin{cor}}
\newcommand{\ecor}{\end{cor}}
\newcommand{\beq}{\begin{equation}}
\newcommand{\eeq}{\end{equation}}
\newcommand{\beqn}{\begin{equation*}}
\newcommand{\eeqn}{\end{equation*}}
\newcommand{\bea}{\begin{eqnarray}}
\newcommand{\eea}{\end{eqnarray}}
\newcommand{\bean}{\begin{eqnarray*}}
\newcommand{\eean}{\end{eqnarray*}}
\newcommand{\ben}{\begin{enumerate}}
\newcommand{\een}{\end{enumerate}}
\newcommand{\bdefn}{\begin{defn}}

\renewcommand\footnotemark{}

\begin{document}
\sloppy
\title{Bounds on the Rate and Minimum Distance of Codes with Availability}

\author{
\IEEEauthorblockN{S. B. Balaji and P. Vijay Kumar, \it{Fellow}, \it{IEEE}}

\IEEEauthorblockA{Department of Electrical Communication Engineering, Indian Institute of Science, Bangalore.  \\ Email: balaji.profess@gmail.com, pvk1729@gmail.com} 

\thanks{P. Vijay Kumar is also an Adjunct Research Professor at the University of Southern California.  This research is supported in part by the National Science Foundation under Grant 1421848 and in part by an India-Israel UGC-ISF joint research program grant.} 
\thanks{S. B. Balaji would like to acknowledge the support of TCS research scholarship program.}
}
\maketitle

	\begin{abstract}
		In this paper we investigate bounds on rate and minimum distance of codes with $t$ availability. We present bounds on minimum distance of a code with $t$ availability that are tighter than existing bounds. For bounds on rate of a code with $t$ availability, we restrict ourselves to a sub-class of codes with $t$ availability called codes with strict $t$ availability and derive a tighter rate bound. Codes with strict $t$ availability can be defined as the null space of an $(m \times n)$ parity-check matrix $H$, where each row has weight $(r+1)$ and each column has weight $t$, with intersection between support of any two rows atmost one. We also present two general constructions for codes with $t$ availability.
	\end{abstract}



	\section{Introduction}
	Let $\mathcal{C}$ denote a linear $[n,k]$ code. The code is $\mathcal{C}$ said to have locality $r$ if each of the $n$ code symbols of $\mathcal{C}$ can be recovered by accessing at most $r$ other code symbols. Equivalently, there exist $n$ codewords ${\underline{h}_1 \cdots \underline{h}_n}$ in the dual code $\mathcal{C}^\perp$ such that $i \in \text{supp}(\underline{h}_i)$ and $|\text{supp}(\underline{h}_i)| \leq r+1$ for $1 \leq i \leq n$ where  $\text{supp}(\underline{h}_i)$ denotes the support of the codeword $\underline{h}_i$.
	
\paragraph{\textbf{Codes with Availability}}
Let $\mathcal{C}$ denote a linear $[n,k]$ code over $\mathbb{F}_q$. $\mathcal{C}$ is said to be a code with $t$ availability if for every code symbol $c_i$ in $\mathcal{C}$, there exist $t$ codewords $\underline{h}^i_1,...,\underline{h}
^i_t$ in the dual code, each of Hamming weight $\leq r+1$, such that $\text{supp}(\underline{h}^i_g) \cap \text{supp}(\underline{h}^i_j) = \{i\} $, $\forall\  1 \leq g \neq j \leq t$. We denote the matrix with these $\underline{h}^j_i$, $ \forall\ 1 \leq j \leq n, 1\leq i \leq t$ as rows by $H_{des}(\mathcal{C})$ (local parity check matrix of $\mathcal{C}$).

The parameter $r$ is called the locality parameter and we will refer to this class of codes as $(n,k,r,t)_{a}$ codes. When the parameters $n,k,r,t$ are clear from the context, we will simply term the code as a code with $t$ availability. 

\paragraph{\textbf{Codes with Strict Availability}}
Codes in this class can be defined as the null space of an $(m \times n)$ parity-check matrix $H$ over a field $\mathbb{F}_q$. In $H$ each row has weight $(r+1)$ and each column has weight $t$, with $nt=m(r+1)$.  Additionally, if the support sets of the rows in $H$ having a non-zero entry in the $i^\text{th}$ column are given respectively by $S^{(i)}_j, j=1,2,\cdots t$, then we must have that $S^{(i)}_j \cap S^{(i)}_l =   \{i \} , \forall\ 1 \leq j \neq l \leq t$.
Thus each code symbol $c_i$ is protected by a collection of $t$ orthogonal parity checks each of weight $(r+1)$. The parameter $r$ is called the locality parameter and we will formally refer to this class of codes as $(n,k,r,t)_{sa}$ codes.  Again, when the parameters $n,k,r,t$ are clear from the context, we will simply term the code as a code with strict $t$ availability.
We will focus only on $(n,k,r,t)_{sa}$ codes for rate bound calculations. We note here that in \cite{WanZhaLiu_Arb_Locality} the authors describe a construction of the most general known availability codes with rate $\frac{r}{r+t}$; these codes also belong to the class of $(n,k,r,t)_{sa}$ codes. Therefore there is evidence to suggest that if $(n,k,r,t)_{sa}$ codes are constructable then they are good candidates for best possible availability codes in terms of rate. Discussion of strict availability codes can be found in a recent paper \cite{CalderBankKadhe}.
	\subsection{Background}
	In \cite{GopHuaSimYek} Gopalan et al. introduced the concept of codes with locality (see also \cite{PapDim,OggDat}), where an erased code symbol is recovered by accessing a small subset of other code symbols. The size of this subset denoted by $r$ is typically much smaller than the dimension of the code, making the repair process more efficient compared to MDS codes. The authors of \cite{GopHuaSimYek} considered codes that can locally recover from single erasures. (see also \cite{GopHuaSimYek,HuaChenLi,KamPraLalKum,TamBar_Optimal_LRC})
	
	Approaches for local recovery from multiple erasures can be found in \cite{PraLalKum,KamPraLalKum,SonDauYueLi,WanZhaLiu_Arb_Locality,WangZhanLin,ZhaWanGe,HuaYaaUchSie,TamBarFro,WanZha_Combinatorial_Repair_locality,TamBar_Optimal_LRC,JuaHolOgg,RawPapDimVis_arxiv}. In this paper we concentrate on codes with $t$ availability and codes with strict $t$ availability.
	\subsection{Our Contributions} 
	Our contributions in this paper include the following:
	\ben[1.]
	\item In subsection \ref{SubSecGreedyAlgo} we derive a bound Eq.\eqref{G115} on rate of an $(n,k,r,3)_{sa}$ code using a greedy algorithm.
	\item In subsection \ref{SubSecTranspose} we derive a bound Eq.\eqref{rate_11} and \eqref{Rate_2} on rate of an $(n,k,r,t)_{sa}$ code for general $t$ using a simple observation on the parity check matrix of a strict availability code and its transpose. The resulting bounds are tighter than the bound given in \cite{TamBarFro} when applied to codes with strict availability.
	\item In section \ref{SecMinDistBounds} we derive field-size dependent (Eq.\eqref{Mindist_Bound1})and field-size independent bounds Eq.\eqref{Mindist_Bound2} and \eqref{Mindist_Bound3} on minimum distance of an $(n,k,r,t)_{a}$ code. Our bounds are tighter than the bounds in \cite{TamBarFro} and \cite{SonYue_Square_Code}.
	\item Finally in section \ref{SecCodeConstructions} we present two general constructions for codes with $t$ availability over the binary field. Some instances of these general constructions give codes with rate higher than $\frac{r}{r+t}$.
	\een


\section{Bounds on the rate of $(n,k,r,t)_{sa}$ codes}\label{SecRateBound}
In \cite{TamBarFro} the following rate-bound was given:

\bea
\text{If $\mathcal{C}$ is an $(n,k,r,t)_a$ code then,} \notag \\
\frac{k}{n} \leq \frac{1}{\prod_{j = 1}^{t}(1+\frac{1}{jr})}. \label{TamoBargRate}
\eea
We compare our rate-bounds with the above bound.
\subsection{A Rate Bound for $(n,k,r,3)_{sa}$ Codes}\label{SubSecGreedyAlgo}
We present a greedy algorithm and analyse the algorithm to get a bound on the rate of an $(n,k,r,3)_{sa}$ code. \\
Let $\mathcal{C}$ be an $(n,k,r,3)_{sa}$ code over a field $\mathbb{F}_q$. Without loss of generality the Tanner graph of $\mathcal{C}$ is assumed to be connected. If not, we can puncture the code and take a subset of code symbols and form a code $\mathcal{C'}$ with rate $\geq$ rate of $\mathcal{C}$ such that $\mathcal{C'}$ is also an $(n',k',r,3)_{sa}$ code with a connected Tanner graph for some $n',k'$ with $\frac{k}{n} \leq \frac{k'}{n'}$ and we can apply the following theorem on $\mathcal{C'}$ which has connected tanner graph.
\begin{thm}
	The code $\mathcal{C}$ which is an $(n,k,r,3)_{sa}$ code over the field $\mathbb{F}_q$ having connected Tanner graph, has rate upper bounded by the following expression:
	\bea
	\frac{k}{n} &\leq&  1-\frac{3(1+L_1+L_2)}{(r+1)(3+L_1+2L_2)}, \label{G115} 
	\eea
	\bean
	\text{where: } m=\frac{3n}{r+1}, \ \ L'_1 = \left \lceil \frac{(2r-1)m}{3(r+2)}-\frac{1}{r+2}-1 \right \rceil, \\
	L_2=\left \lfloor \frac{m-3-L'_1}{2} \right \rfloor,  L_1= m-3-2L_2.
	\eean
\end{thm}
\begin{proof} We present and analyse a greedy algorithm. By definition, $\mathcal{C}$ is the null space of an $m \times n$ matrix $H$ which contains in its rows, all the orthogonal parities protecting all the symbols with $m(r+1)=3n$. Let $\text{Rows}(H)=$ the set of row vectors of the matrix $H$.\\\\ 
	\textbf{Greedy Algorithm:} \\
	\ben
	\item Let $S= \emptyset, P=\emptyset$. 
	\item Step 1: Pick an arbitrary number $\sigma_1$ from $[n]$ and set $S=\{\sigma_1\}$ and $P = \{\underline{c} \in \text{Rows}(H) : \sigma_1 \in \text{Support}(\underline{c}) \}$.
	\item Step $i$, $i \geq 2$: Choose a number $\sigma_i \in [n]-S$ such that
	\bean
	  \sigma_i = \text{argmax}_{\{j \in [n]-S \}} (|D_j| \times I(|D_j|\leq 2))
	 \eean
	where $D_j = \{ \underline{c} \in P : j \in \text{Support}(\underline{c}) \}$ and $I(|D_j|\leq 2)$ is an indicator function which is 1 if $|D_j|\leq 2$ and 0 otherwise. If there are multiple argument values attaining the maximum, choose one of them randomly and assign it to $\sigma_i$.
	Now $S=S \cup \{\sigma_i\}$ and $P = P \cup \{\underline{c} \in \text{Rows}(H)-P : \sigma_i \in \text{Support}(\underline{c}) \}$.
	\item \textbf{Pseudocode for the Greedy Algorithm:}\\
	Set $P=\emptyset$\\
	Set $S=\emptyset$ \\
	Set $i=1$\\
	while $|P| <m$ \\ 
	$\indent \indent$ execute Step $i$ \\
	$\indent \indent$ $i=i+1$ \\
	end while.
	\item It is clear that, $k \leq n-|S|$ at the end of the algorithm. \\
	\een
	\textbf{Analysis of the Greedy Algorithm:}
	
	In the below arguments, wherever we write S it refers to the set S at the end of the greedy algorithm (i.e., when $|P|=m$). Let $g_i$ be the number of new codewords added to $P$ at step $i$. Since the Tanner graph of $\mathcal{C}$ is connected, $g_i \in \{1,2\}$ for $2 \leq i \leq |S|$. We define $g_{|S|+1}=0$ at the end of the algorithm. Let $P=\{\underline{c}_1,..,\underline{c}_f\}$ after the step $i$. Now define a partial parity check matrix $H^i_{par}=[\underline{c}_1^T,...,\underline{c}_f^T]^T$ (we refer to the codeword $\underline{c}_i$ as a row vector). Let $s^i_1,s^i_2,s^i_3$ be the number of weight $1,2,3$ columns in $H^i_{par}$ respectively. Let us introduce four collections of variables $I_i,\phi_i,J_i,\psi_i$ for each $1 \leq i \leq |S|-1$.\\
	Keeping in mind the properties of the greedy algorithm, we have the following recursive update for $i \geq 2$:
	
 	\bean
 	\text{if } g_{i+1}&=&2,g_{i+2}=2 \text{ then}\\
 	s^{i+1}_1 &=& s^i_1 + 2r-1 \\
 	s^{i+1}_2 &=& s^i_2 + 0 \text{ \ \  ($\Leftarrow$ $s^i_2=0$ because $g_{i+1}=2$ and $s^{i+1}_2=0$ because $g_{i+2}=2$) } \\
 	s^{i+1}_3 &=& s^i_3 + 1 \\ \\
 	\text{if } g_{i+1}&=&1,g_{i+2}=2 \text{ then}\\
 	s^{i+1}_1 & =& s^i_1 - \phi_i + r+1 \\
 	s^{i+1}_2  &=& s^i_2 - \phi_i \text{ \ \  ($\Leftarrow$ $s^{i+1}_2 = 0$ because $g_{i+2}=2$) } \\
 	s^{i+1}_3 &=& s^i_3 + \phi_i \text{ \ \  ($\Leftarrow$ $\phi_i$ two weight columns gets converted to three weight columns)} \\
 	\text{for some } 0 & \leq & \phi_i \leq r+1 \text{ \ \  ($\Leftarrow$ column indices corresponding to $\phi_i$ three weight columns, added in step $i+1$, corresponds to } \\ 
 	& & \ \  \ \ \ \ \ \ \ \ \ \ \ \ \ \text{ a subset of the support of the codeword added to $P$ at step $i+1$ ($g_{i+1}=1$). Hence $\phi_i \leq r+1$)} \\ \\
 	\text{if } g_{i+1}&=&2,g_{i+2}=1 \text{ then}\\
 	s^{i+1}_1 &=& s^i_1 + 2r-1-2I_i \\
 	s^{i+1}_2 &=& s^i_2 + I_i  \text{ \ \ ($\Leftarrow$ $s^i_2=0$ because $g_{i+1}=2$ )} \\
 	s^{i+1}_3 &=& s^i_3 + 1  \\
 	\text{for some } 0 & < & I_i \leq 2r \\ \\ 
 	\text{if } g_{i+1}&=&1,g_{i+2}=1 \text{ then}\\
 	s^{i+1}_1 & =& s^i_1 + r+1 -J_i -2 \psi_i \\
 	s^{i+1}_2  &=& s^i_2-J_i+\psi_i  \text{ \ \  ($\Leftarrow$ $\psi_i$ new two weight columns gets added and} \\ 
 	\hspace{1cm}
 		& & \ \ \ \ \ \ \ \ \ \ \ \ \ \ \ \ \ \ \text{ $J_i$ already existing two weight columns gets converted to three weight columns)}\\
 	s^{i+1}_3 &=& s^i_3 + J_i \text{ \ \ \ \ \ \ \ \  ($\Leftarrow$ $J_i$ two weight columns gets converted to three weight columns)} \\
 	\text{for some } s^i_2 & \geq & J_i \geq 1, r+1 \geq J_i+\psi_i \geq 0 \\ \\
 	\text{if } g_{i+1}&=&1,g_{i+2}=0 \text{ then}\\
 	s^{i+1}_1 & =& s^i_1 + 0 \\
 	s^{i+1}_2  &=& s^i_2 - (r+1) \text{\ \ \ ($\Leftarrow$ final step )}  \\
 	s^{i+1}_3 &=& s^i_3 + r+1
 	\eean
 	Writing the first two steps of the update explicitly: 
 	\bean
 	s^{1}_1 & =& 3r \\
 	s^{1}_2  &=& 0 \\
 	s^{1}_3 &=& 1
 	\eean
 	\bean
 	s^{2}_1 &=& 3r+2r+2-3-2 \gamma_1 \\
 	s^{2}_2  &=& \gamma_1 \\
 	s^{2}_3 &=& 2 \\
 	\text{for some } 0 & \leq & \gamma_1 \leq 4
 	\eean
 	
 		Let $l_{kj}=|\{i+1 : g_i=k,g_{i+1}=j,|S| \geq i \geq 3\}|$ and $S_{kj}=\{i : g_{i+1}=k,g_{i+2}=j,|S|-1 \geq i \geq 2\}$ at the end of the algorithm. We set $I_i=0, \forall i \notin S_{21}$. We set $\phi_i=0, \forall i \notin S_{12}$. We set $J_i=0, \forall i \notin S_{11}$. We set $\psi_i=0, \forall i \notin S_{11}$. We note that $l_{kj}=|S_{kj}|$\\
 	Now using the Global constriants ($s^{|S|}_1=0,s^{|S|}_2=0,s^{|S|}_3=n$ at the end of the algorithm i.e., after the final step all the columns in the partial parity check matrix must have weight 3 (Hence partial parity check matrix after the final step is the full matrix H with rows permuted.)):\\
 	
 	\bea
 	\hspace{-5cm}
 	s_3 &=& 2+l_{22} + l_{21} + \sum_{i \in S_{11}} J_i + \sum_{i \in S_{12}} \phi_i + r+1 =n \label{G3} \\
 	s_2 &=& \gamma_1 - \sum_{i \in S_{12}} \phi_i + \sum_{i \in S_{21}} I_i -\sum_{i \in S_{11}} J_i + \sum_{i \in S_{11}} \psi_i - (r+1) = 0 \label{G2} 
 	\eea
 	\bea
 	s_1 &=& 5r-1-2\gamma_1 + l_{22} (2r-1) + l_{21} (2r-1) -\notag\\
 	& &  2 \sum_{i \in S_{21}} I_i + l_{11} (r+1) - \sum_{i \in S_{11}} J_i -2 \sum_{i \in S_{11}} \psi_i + \notag \\
 	& & l_{12} (r+1) - \sum_{i \in S_{12}} \phi_i = 0 \label{G1}
 	\eea
 	\bea
 	l_{11} &\geq& \sum_{\{i: \ \ g_{i+1}=2,g_{i+2}=1, \ \ 2 \leq i \leq |S|-1 \}} \frac{I_i-2}{2} \label{G4} \\
 	l_{21} -1 &\leq & l_{12} \leq l_{21} \label{G5} \\
 	m &=& \frac{3n}{r+1} =   5+g_3+2(l_{22} + l_{12})+l_{21} + l_{11} \label{G6}
 	\eea
			Equation \eqref{G6} is true because we are counting the second index of $l_{kj}$ for codewords in the dual added to $P$. This is the reason we have $i+1$ in the definition of $l_{kj}$ instead of $i$ (putting $i$ in definition of $l_{kj}$ would have given the same value of $l_{kj}$). \\
			Inequality \eqref{G4} is true because whenever $g_{i+1}=2,g_{i+2}=1, i \geq 2$, it is followed by at least $z=\frac{I_i-2}{2}$, (1,1) transitions i.e., $g_{i+2+j}=1,g_{i+2+j+1}=1, \forall 0 \leq j \leq z-1$. This is because at least $2(z+1)$ new 2 weight columns appears in the partial parity check matrix $H^{i+1}_{par}$ compared to $H^{i}_{par}$. Among the coordinates (column indices) corresponding to these new $2(z+1)$ two weight columns at least a subset of size $z+1$ corresponds to a subset of support of one of the new codeword added in step $i+1$ to the set $P$ and hence these $z+1$ new 2 weight columns can become 3 weight columns only one by one at each of the following steps as the intersection between support of any two codewords in $H$ is atmsot one (since there is 2 weight column in each of these steps, we will be adding only one codeword to $P$ in each of these steps).
			Let,
			\bean
			\sum_{i \in S_{11}} J_i = \sum_{\{i:  \ \ g_{i+1}=1,g_{i+2}=1, \ \  2 \leq i \leq |S|-1 \}} J_i & =& l_{11} \frac{\sum_{\{i:\ \ g_{i+1}=1,g_{i+2}=1, \ \  2 \leq i \leq |S|-1  \}} J_i}{l_{11}} = l_{11} J \\
			\sum_{i \in S_{11}} \psi_i = \sum_{\{i: \ \ g_{i+1}=1,g_{i+2}=1, \ \  2 \leq i \leq |S|-1 \}} \psi_i &= & l_{11} \frac{\sum_{\{i:\ \ g_{i+1}=1,g_{i+2}=1, \ \  2 \leq i \leq |S|-1 \}} \psi_i}{l_{11}} = l_{11} \psi \\	
				\sum_{i \in S_{21}} I_i = \sum_{\{i: \ \ g_{i+1}=2,g_{i+2}=1, \ \  2 \leq i \leq |S|-1 \}} I_i & =& l_{21} \frac{\sum_{\{i: \ \ g_{i+1}=2,g_{i+2}=1,  \ \ 2 \leq i \leq |S|-1 \}} I_i}{l_{21}} = l_{21} I \\				
					\sum_{i \in S_{12}} \phi_i = \sum_{\{i: \ \  g_{i+1}=1,g_{i+2}=2, \ \  2 \leq i \leq |S|-1 \}} \phi_i & =& l_{12} \frac{\sum_{\{i: \ \ g_{i+1}=1,g_{i+2}=2, \ \  2 \leq i \leq |S|-1 \}} \phi_i}{l_{12}} = l_{12} \phi \\
						\text{Now } 0 & \leq & J+\psi \leq r+1 \\	
						0 & \leq & \phi \leq r+1 \\	
						0 & < & I \leq 2r.
			\eean
			\bea
						\text{Now \eqref{G4}} \text{ becomes: } \notag \\
						2 l_{11}+2l_{21}  &\geq& l_{21} I \label{G12}
			\eea
				\ben
				\item Manipulating the equations and inequalitie given above: \\
				
				 \bea
				  \text{By equation \eqref{G2}} &:& \notag\\
				  l_{12} \phi + l_{11} J + (r+1) - \gamma_1 & = & l_{21} I + l_{11} \psi \label{New_6} \\
				  \text{Substituting } r+1 & \geq & J+\psi \text{ in } \eqref{New_6}:  \notag \\
				  l_{12} \phi + l_{11} (r+1-\psi) + (r+1) - \gamma_1 & \geq & l_{21} I + l_{11} \psi  \notag \\
				  \frac{l_{12} \phi - l_{21} I + l_{11} (r+1) + (r+1) - \gamma_1}{2} & \geq & l_{11} \psi \label{G_11}
				 \eea

				 \bea
			      \text{Substituting \eqref{G3} in \eqref{G2}} &:& \notag \\
				 0 & =& \gamma_1 + l_{21} I + l_{11} \psi - (r+1) + 2 + l_{22} + l_{21} + r+1 -n \notag \\
				 l_{22} & =& n-\gamma_1-2-l_{21}-l_{21} I - l_{11} \psi  \label{New_1} \\
				 \text{Substituting  \eqref{G_11} in equation \eqref{New_1} } &:& \notag \\
				 l_{22} & \geq & n-\frac{\gamma_1}{2}-2-l_{21}-\frac{l_{21} I}{2}-\frac{l_{12} \phi}{2} \notag \\
				 & & - \frac{l_{11} (r+1)}{2} - \frac{(r+1)}{2} \label{New_2} \\
				 \text{Substituting \eqref{G12} in inequality \eqref{New_2}} &:& \notag \\
				 l_{22} & \geq & n-\frac{\gamma_1}{2}- \frac{(r+1)}{2}- 2-l_{21}-l_{11}-l_{21}- \notag \\ 
				 & & \frac{l_{12} \phi}{2} -\frac{l_{11} (r+1)}{2}  \label{New_3} \\
				 \text{Rewriting \eqref{G6} and substituting the inequality \eqref{New_3} in \eqref{G6}} &:& \notag\\
				 m = \frac{3n}{r+1} &=&   5+g_3+2(l_{22} + l_{12})+l_{21} + l_{11} \notag \\
				 m & \geq & 5+g_3+ 2n-\gamma_1-(r+1)- 4-2 l_{21}-2l_{11}-2l_{21}-l_{12} \phi \notag  \\
				 & &  - l_{11} (r+1) + 2l_{12}+l_{21} + l_{11} \notag  \\ 				 
				 0 & \geq & 5+g_3+ 2n-m-\gamma_1-(r+1)- 4-3 l_{21} \notag  \\
				 & & -l_{12} (\phi-2)- l_{11} (r+2) \label{New_4} 
				 \eea
				  Now first substituting $\phi \leq r+1$ and then substituting $l_{12} \leq l_{21}$ in the inequality \eqref{New_4} :
				  \bea
				 0 & \geq & 5+g_3+ 2n-m-\gamma_1-(r+1)- 4 \notag \\
				 & & -l_{21} (r+2)- l_{11} (r+2) \label{New_5} 
				 \eea
				 \bea
                 l_{21} + l_{11} & \geq & \frac{5+g_3+ 2n-m-\gamma_1-(r+1)- 4}{r+2} \notag \\
                 l_{21} + l_{11} & \geq & \frac{(2r-1)m}{3(r+2)}+\frac{5+g_3-\gamma_1-(r+1)- 4}{r+2} \label{New_7} \\
                  \text {Substituting $g_3\geq 1$, $\gamma_1 \leq 4$ in \eqref{New_7}} &:& \notag \\
                  l_{21} + l_{11} & \geq & \frac{(2r-1)m}{3(r+2)}-\frac{(r+3)}{r+2} \notag
				 \eea
				
				\bea
				\text{which gives } & & \notag\\
				l_{11} + l_{21} & \geq & \frac{(2r-1)m}{3(r+2)}-\frac{1}{r+2}-1 \label{G14}
				\eea
				
				\item  For $j = 1,2$, let $L_j=|\{i : g_i = j, |S| \geq i \geq 1\}|$ at the end of the algorithm. Using $|S|=L_1+L_2+1$ and $m=L_1+2L_2+3$ we get:\\
				\bea
				\frac{k}{n} &\leq& 1-\frac{|S|}{n} = 1-\frac{1+L_1+L_2}{n} \notag \\
				\frac{k}{n} &\leq& 1-\frac{m(1+L_1+L_2)}{nm} = 1-\frac{3(1+L_1+L_2)}{(r+1)(3+L_1+2L_2)} \label{G15} 
				\eea
				\item Now using \eqref{G14} we get :\\
				\bea
				L_1 \geq l_{11}+l_{21}  \geq \left \lceil \frac{(2r-1)m}{3(r+2)}-\frac{1}{r+2}-1 \right \rceil \label{New_8}
				\eea
				Using the inequality \eqref{New_8} on $L_1$ and $m=L_1+2L_2+3$, we have:
				\bea
				L'_1 = \left \lceil \frac{(2r-1)m}{3(r+2)}-\frac{1}{r+2}-1 \right \rceil \notag \\
				L_2 \leq \left \lfloor \frac{m-3-L'_1}{2} \right \rfloor \label{G116} \\
				L_1 \geq m-3-2 \left \lfloor\frac{m-3-L'_1}{2} \right \rfloor. \label{G117}
				\eea
				Substituting the bounds \eqref{G116},\eqref{G117} in \eqref{G15} we get the bound given in the theorem.
				\een
	
\end{proof}

Fig.~\ref{Rate_3} shows the plot of the new rate bound \eqref{G115} along with the rate bounds given in \cite{TamBarFro} ($\frac{k}{n} \leq \frac{1}{\prod_{j=1}^{3}(1+\frac{1}{jr})}$) and \cite{SonYue_3_Erasure} ($\frac{k}{n} \leq \frac{r^2}{(r+1)^2}$) for $t=3,n={r+3 \choose 3}$. It can be seen that the new bound given in \eqref{G115} is tighter, although the tightness is in the restricted setting of strict availability. Even if we plot for different $n$, the new bound remains tighter barring some small values of $r$ as low as $r \leq 3$. The fact that the new rate bound depends on $n$ is by itself interesting and such bounds might throw insight on optimal value of $n$ for $(n,k,r,3)_{sa}$ codes with connected tanner graph. We note that the achievability curve plotted corresponds to a construction taken from \cite{WanZhaLiu_Arb_Locality} which corresponds to a code having strict availability.

\begin{figure}[h!]
	\centering
	\includegraphics[width=5in]{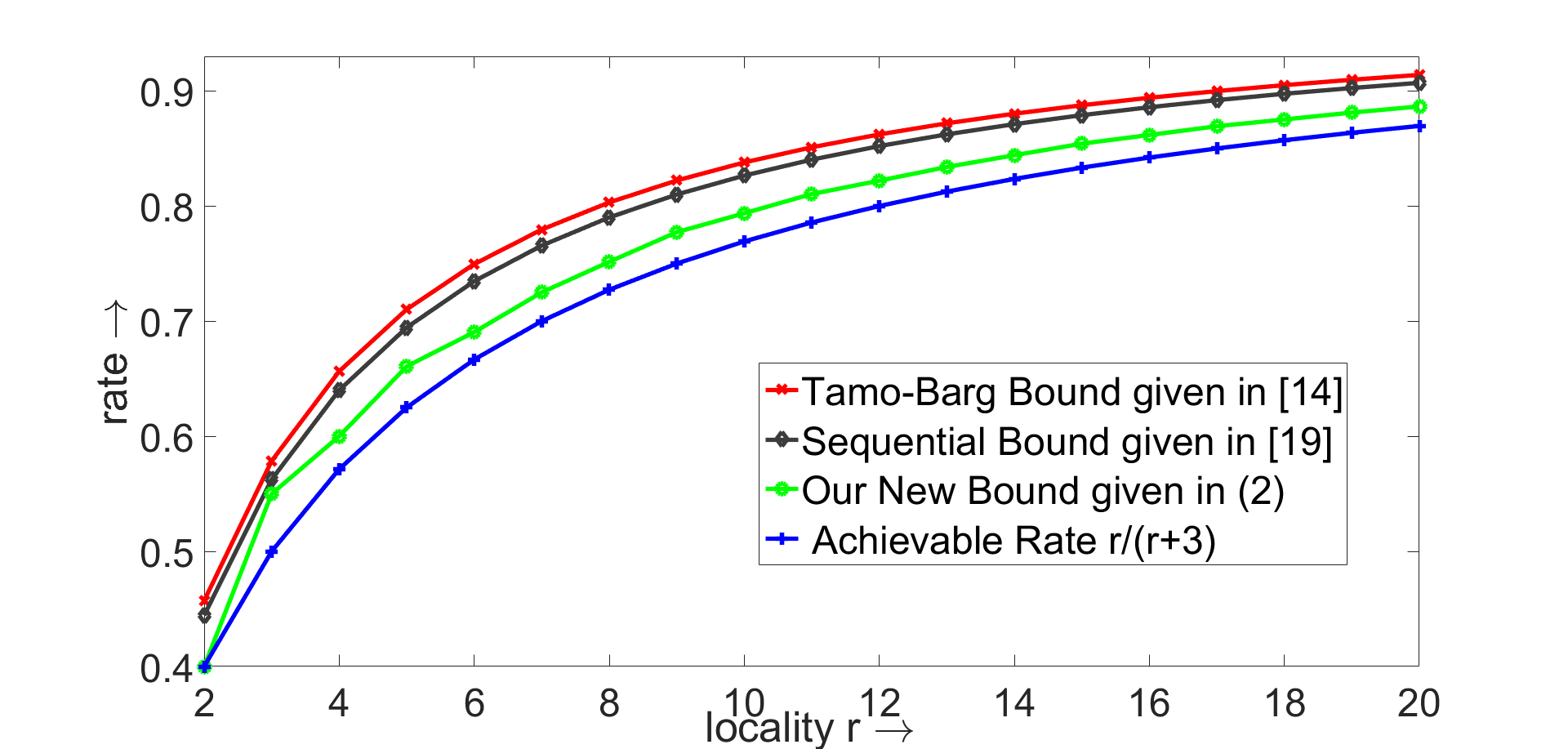}
	\caption[Example plot of locality vs rate]{Plotting locality vs rate for $t=3,n={r+3 \choose 3}$. Here we compare the bound \eqref{G115} with the bounds given in \cite{TamBarFro} and \cite{SonYue_3_Erasure} and the achievable rate $\frac{r}{r+3}$ given by the construction in \cite{WanZhaLiu_Arb_Locality}. Note that our bound \eqref{G115} is tighter than the bound given in \cite{TamBarFro} and \cite{SonYue_3_Erasure}.}
	\label{Rate_3}
\end{figure}

\subsection{A simple rate bound for an $(n,k,r,t)_{sa}$ code:}\label{SubSecTranspose}
Here we derive a bound on rate of an $(n,k,r,t)_{sa}$ code using a very simple transpose trick.
\begin{thm}
	 
	 \bea
	   \text{Let }R(r,t) & = &  sup_{\{(k,n) : (n,k,r,t)_{sa} \text{ code exists over some field $\mathbb{F}_q$} \}} \frac{k}{n}. \notag \\
	 \text {Then, }R(r,t) &=& 1-\frac{t}{r+1}+\frac{t}{r+1} R(t-1,r+1)) \label{rate_11} \\
	 R(r,t) &\leq& 1-\frac{t}{r+1}+\frac{t}{r+1} \frac{1}{\prod_{j=1}^{r+1}(1+\frac{1}{j(t-1)})} \label{Rate_2}
	 \eea
	 
\end{thm}
\begin{proof}
Let $R^{(n,q)}(r,t)$ be the maximum achievable rate of a code with strict $t$ availability for fixed $n,r,t$ over the field $\mathbb{F}_q$. If $(n,k,r,t)_{sa}$ code doesn't exist for any $k$ for fixed $n,r,t,q$ then we define $R^{(n,q)}(r,t)=-\infty$. Let us choose $n,r,t,q$ such that $R^{(n,q)}(r,t)>0$. Let $\mathcal{C}$ be an $(n,k,r,t)_{sa}$ code over the field $\mathbb{F}_q$ with rate $R^{(n,q)}(r,t)$. By definition $\mathcal{C}$ is the null space of an $m \times n$ matrix $H$ with all columns having weight $t$ and all rows having weight $r+1$. This matrix $H$ contains all $t$ orthogonal parity checks protecting any given symbol. Now the null space of $H^T$ (transpose of $H$) corresponds to an $(m,k'',t-1,r+1)_{sa}$ code over the field $\mathbb{F}_q$. Hence we have the following inequality:
\bean
rank(H)=n(1-R^{(n,q)}(r,t)) \\
rank(H)=rank(H^T) \geq m(1-R^{(m,q)}(t-1,r+1)) \\
\eean
Hence we have
\bean
m(1-R^{(m,q)}(t-1,r+1)) & \leq & n(1-R^{(n,q)}(r,t)) \\
\text{Using } m(r+1)=nt : \\
\frac{t}{r+1}(1-R^{(m,q)}(t-1,r+1)) & \leq & (1-R^{(n,q)}(r,t)) \\
R^{(n,q)}(r,t) & \leq & 1-\frac{t}{r+1}+\frac{t}{r+1}R^{(m,q)} (t-1,r+1) \\
R^{(n,q)}(r,t) & \leq & 1-\frac{t}{r+1}+\frac{t}{r+1} (sup_{\{ m \geq 0, q=p^w : \text{$p$ is a prime and } w \in \mathbb{Z}_+ \}} R^{(m,q)}(t-1,r+1)) \\
sup_{\{ n \geq 0, q=p^w : \text{$p$ is a prime and } w \in \mathbb{Z}_+ \}} R^{(n,q)}(r,t) & \leq & 1-\frac{t}{r+1}+\frac{t}{r+1} (sup_{\{ m \geq 0, q=p^w : \text{$p$ is a prime and } w \in \mathbb{Z}_+ \}} R^{(m,q)}(t-1,r+1)) \\
R(r,t) & \leq & 1-\frac{t}{r+1}+\frac{t}{r+1} R(t-1,r+1)) 
\eean
Now swapping the roles of $H$ and $H^T$ in the above derivation i.e., we take an $(m,k',t-1,r+1)_{sa}$ code $\mathcal{C}$ over the field $\mathbb{F}_q$ with rate $R^{(m,q)}(t-1,r+1)$ and repeat the above argument in exactly the same way. By doing so we get :\\
\bean
R(r,t) \geq 1-\frac{t}{r+1}+\frac{t}{r+1} R(t-1,r+1) 
\eean

Hence we get:\\
\bea
R(r,t) = 1-\frac{t}{r+1}+\frac{t}{r+1} R(t-1,r+1) \label{rate_111}
\eea

Now substituting the rate bound  $R(t-1,r+1)) \leq \frac{1}{\prod_{j=1}^{r+1}(1+\frac{1}{j(t-1)})}$  given in \cite{TamBarFro} into \eqref{rate_111}, we get:

\bea
R(r,t) \leq 1-\frac{t}{r+1}+\frac{t}{r+1} \frac{1}{\prod_{j=1}^{r+1}(1+\frac{1}{j(t-1)})} 
\eea
\end{proof}
\begin{figure}[h!]
	\centering
	\includegraphics[width=5in]{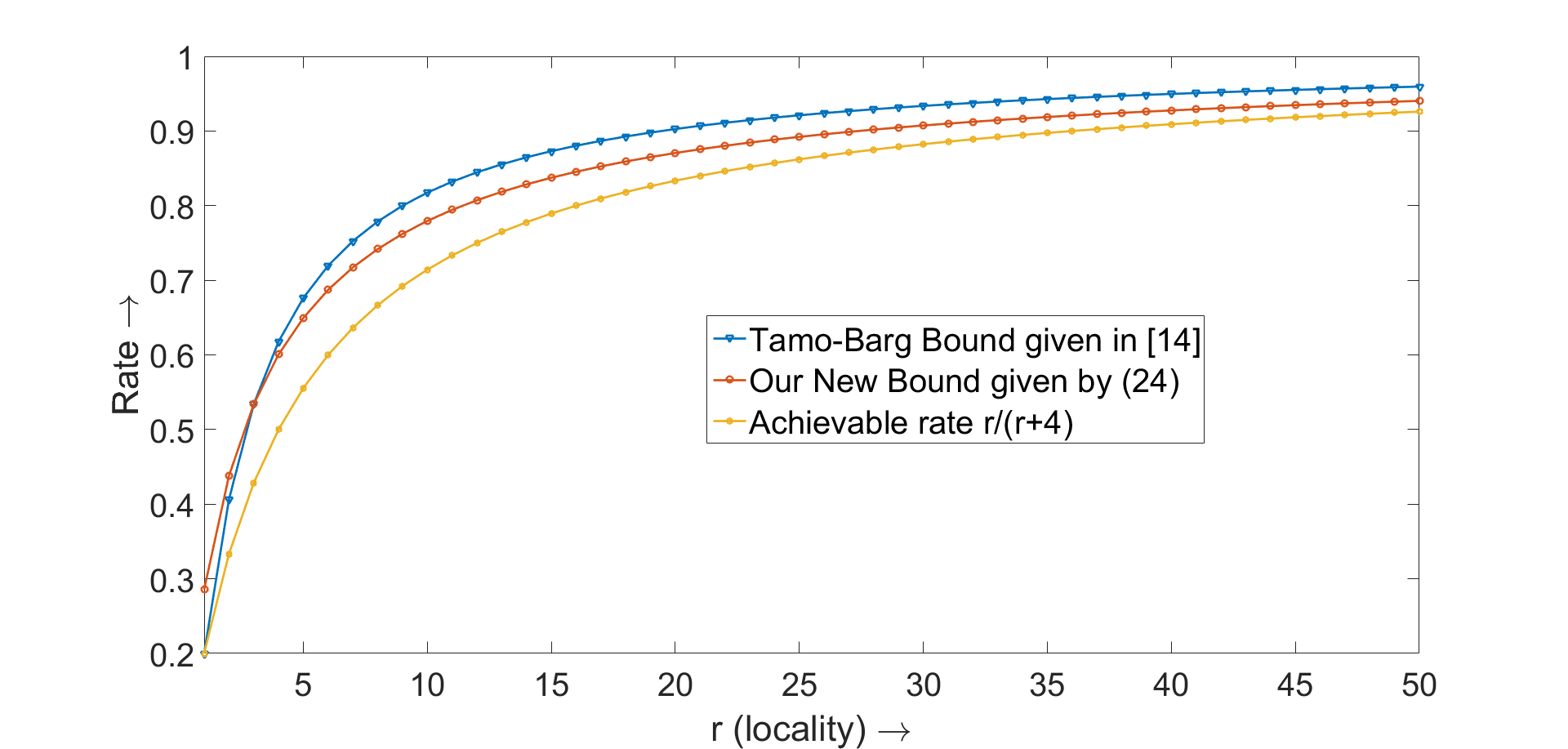}
	\caption[Example plot of locality vs rate]{Plotting locality vs rate for $t=4$. Here we compare the bound \eqref{Rate_2} with the bound given in \cite{TamBarFro} and the achievable rate $\frac{r}{r+4}$ given by the construction in \cite{WanZhaLiu_Arb_Locality}.  Note that our bound \eqref{Rate_2} is tighter than the bound given in \cite{TamBarFro}.}
	\label{Rate_4}
\end{figure}
\begin{note} $\bf{Tightness \ of \ the \ bound }$:\\
The bound on $R(r,t)$ given in \eqref{Rate_2} becomes tighter than the bound $R(r,t) \leq \frac{1}{\prod_{j=1}^{t}(1+\frac{1}{jr})}$ given in \cite{TamBarFro} as $r$ increases for a fixed $t$. As an example lets calculate $R(r,2)$. From \eqref{Rate_2}:

\bean
R(r,2) \leq 1-\frac{2}{r+1}+\frac{2}{r+1} \frac{1}{\prod_{j=1}^{r+1}(1+\frac{1}{j})} \\
R(r,2) \leq 1-\frac{2}{r+1}+\frac{2}{r+1} \frac{1}{r+2} \\
R(r,2) \leq \frac{r}{r+2}.
\eean
The above abound on $R(r,2)$ obtained from  \eqref{rate_11} and  \eqref{Rate_2} is a tight bound as it is known that rate $\frac{r}{r+2}$ is achievable for $t=2$ and any $r$ using a complete graph code (\cite{PraLalKum}) and hence clearly tighter than $R(r,2) \leq \frac{r^2}{(r+1)(r+\frac{1}{2})}$ given in \cite{TamBarFro}. 
We show a plot of the bound given in \eqref{Rate_2} for $t=4$ in  Fig.~\ref{Rate_4}. The plot shows that the bound given by \eqref{Rate_2} is tighter than the bound given in \cite{TamBarFro} for $t=4, r>2$. \\
Even though our bound becomes tighter as $r$ increases for a fixed $t$ than the bound in \cite{TamBarFro}, the bound given in \cite{TamBarFro} is for $(n,k,r,t)_{a}$ codes but the bound in \eqref{rate_11} and  \eqref{Rate_2} is applicable only for $(n,k,r,t)_{sa}$ codes. But we also would like to draw attention to the fact that most of the high rate constructions known in literature for $(n,k,r,t)_{a}$ codes are also $(n,k,r,t)_{sa}$ codes. In fact the most general high rate construction for $(n,k,r,t)_{a}$ is given in \cite{WanZhaLiu_Arb_Locality} and it is also an $(n,k,r,t)_{sa}$ code. Hence there is a very good reason to think that when it comes to rate-optimality $(n,k,r,t)_{sa}$ codes will give good candidate codes.
\end{note}
\subsection{A numerical bound on rate of an $(n,k,r,t)_{sa}$ code}
A technique to bound rate of an $(n,k,r,t)_a$ code with given minimum distance was proposed in \cite{HaoRecursive}. We follow a similar approach to give a numerical bound on rate of an $(n,k,r,t)_{sa}$ code.\\
 We have not defined information symbol availability code in this paper but for our purpose it suffices to note that bounds on dimension and minimum distance for $t$ information symbol availability code also hold for codes with $t$ availability.
A field-size dependent bound is given by Huang et al. in \cite{HuaYaaUchSie} for information symbol availability codes :
For any linear $[n,k,d]_q$ code with $t$ information symbol availability and hence for an $(n,k,r,t)_a$ code with minimum distance $d$, the dimension satisfies:
\begin{align}
k \leq \underset{\substack{{1 \leq x \leq \left \lceil\frac{k-1}{(r-1)t+1}\right \rceil,x \in \mathbb{Z}^+,} \\ {\mathbf{y} \in ([t])^x,} \\ {A(r,x,\mathbf{y}) < k}}}{\text{min}}\{A(r,x,\mathbf{y})+k_{l-opt}^{(q)}[n-B(r,x,\mathbf{y}),d]\}, \label{HuangSiegelRate}
\end{align} 	
where $x$ is a positive integer and $\mathbf{y} = (y_1,...,y_x) \in ([t])^x$ is a vector of $x$ positive integers, $k_{l-opt}^{(q)}[n,d]$ is the largest possible dimension of a linear code over $\mathbb{F}_q$ with block length $n$ and minimum distance $d$ and 
\begin{align}
A(r,x,\mathbf{y}) & = \sum_{j = 1}^{x}(r-1)y_j+x,\\
B(r,x,\mathbf{y}) & = \sum_{j = 1}^{x}ry_j+x.	
\end{align}
In this section, we provide a Linear Programming upper bound on the rate of a binary $(n,k,r,t)_{sa}$ (strict availability) code $\mathcal{C}$ with minimum distance $d_{\text{min}}(\mathcal{C})=d$ and compare the bound for a special case with the bounds Eq.\eqref{TamoBargRate} and Eq.\eqref{HuangSiegelRate}.\\
Let $\mathcal{C}$ be an $(n,k,r,t)_{sa}$ code over a field $\mathbb{F}_q$ and $\mathcal{C}^\perp$ be its dual code.  By defintion, $\mathcal{C}$ is the null space of an $m \times n$ matrix $H$ which contains all the orthogonal parities protecting all the symbols with $m(r+1)=nt$.  Let $A_i$ and $B_i$ be the number of codewords of weight $i$ in $\mathcal{C}$ and $\mathcal{C}^\perp$ respectively. Then due to linearity of the code and its dual we have
\begin{align}
\sum_{i = 0}^{n} A_i & = 1 + \sum_{i = d}^{n}A_i = |\mathcal{C}| = q^k,\\
\sum_{j = 0}^{n} B_j & = |\mathcal{C}^\perp| = q^{n-k}.
\end{align}
From MacWilliam's identity we have,
\begin{align}
B_j = \frac{1}{|\mathcal{C}|}\sum_{i = 0}^{n}A_i\mathcal{K}_j(i), \text{ }1 \leq j \leq n,
\end{align}
where $\mathcal{K}_j(i) = \sum_{a = 0}^{j}(-1)^a(q-1)^{j-a}\binom{i}{a}\binom{n-i}{j-a}$ are Krawtchouk polynomials.\\

For a strict availability code upon equating the sum of the row weights and sum of column weights of the matrix whose rows contain exactly all the codewords in the dual of weight equal to $r+1$ we have:
\begin{align}
(r+1)B_{r+1} \geq nt. \label{LP_Eq1}
\end{align}
For a code with strict $t$ availability, we make the following observation:

For a codesymbol $c_i$, $1 \leq i \leq n$ there are $t$ codewords $h_1^i,...,h_t^i$ in the rows of $H$ with support sets $S_1^{(i)},...,S_t^{(i)}$ $(S_h^{(i)} \cap S_g^{(i)} = \{i\}, \forall 1 \leq h \neq g \leq t)$ as described in the definition of a strict availability code. Now adding any pair of rows $h_x^i$ and $h_y^i$ out of the $\binom{t}{2}$ such combinations will result in dual codewords of weight exactly $2r$. If we do the same for every codesymbol, it can be seen that $n\binom{t}{2}$ distinct $2r$ weight codewords in the dual are generated if $r > 2$.  The fact that these $n\binom{t}{2}$, $2r$ weight codewords in the dual result in distinct codewords is due to the property of strict availability that the support sets of any two codewords in the rows of $H$ can intersect at at most one coordinate. The same property also implies that in the $(m \times n)$ matrix $H$ upon adding every pair of rows we get at most(exact for $r > 2$) $n\binom{t}{2}$ dual codewords of weight $2r$ and $\binom{m}{2} - n\binom{t}{2}$ dual codewords of weight $2(r+1)$.  Thus we have
\begin{align}
B_{2r} & \geq n\binom{t}{2}, \text{ for }r > 2, \label{(B2r)} \\
B_{2(r+1)} & \geq \binom{m}{2} - n\binom{t}{2} = \binom{\frac{nt}{r+1}}{2} - n\binom{t}{2}. \label{B2(r+1)}
\end{align}
We also have $d_{\text{min}}(\mathcal{C}) \geq t+1$ for any availability $t$ code. Upon simplification inequalities \eqref{(B2r)} and \eqref{B2(r+1)} take the form:
\begin{align}
\sum_{i = t+1}^{n}A_i\left(n\binom{t}{2}-\mathcal{K}_{2r}(i)\right) & \leq (q-1)^{2r}\binom{n}{2r}-n\binom{t}{2},  \text{ for }r > 2,\label{B2rsimplified}\\
\sum_{i = t+1}^{n}A_i\left(\binom{\frac{nt}{r+1}}{2}-n\binom{t}{2}-\mathcal{K}_{2(r+1)}(i)\right) & \leq (q-1)^{2(r+1)}\binom{n}{2(r+1)}-\binom{\frac{nt}{r+1}}{2}+n\binom{t}{2}.\label{B2(r+1)simplified}
\end{align}
Also, from the non-negativity constraints on $A_i$ and $B_i$,
\begin{align}
A_i & \geq 0, \text{ for } 1 \leq t+1 \leq n, \label{Ai}\\
\sum_{i = t+1}^{n}A_i\mathcal{K}_j(i) & \geq -(q-1)^j\binom{n}{j},\text{ for }0 \leq j \leq n.\label{Bj}
\end{align}
Now the problem of finding an upper bound on the dimension is one of maximizing the number of codewords in the code $\mathcal{C}$ subject to the mentioned constraints that are linear in $A_{t+1},...,A_n$. Thus we have the following linear program formulation:
\begin{align}
\text{maximize } & 1+\sum_{i = t+1}^{n}A_i \label{LP_11}\\
\text{s.t. } & \eqref{Ai},\eqref{Bj},\eqref{B2rsimplified},\eqref{B2(r+1)simplified},\eqref{LP_Eq1} \text{ hold} \label{LP_12}.
\end{align}
If $M$ is the maximum value of the objective function of the linear program described above then we have the following bound on the dimension of the code:
\begin{align}
k \leq \log_q(M)
\end{align}
\begin{figure}[ht]
	\centering
	\includegraphics[scale = 0.3]{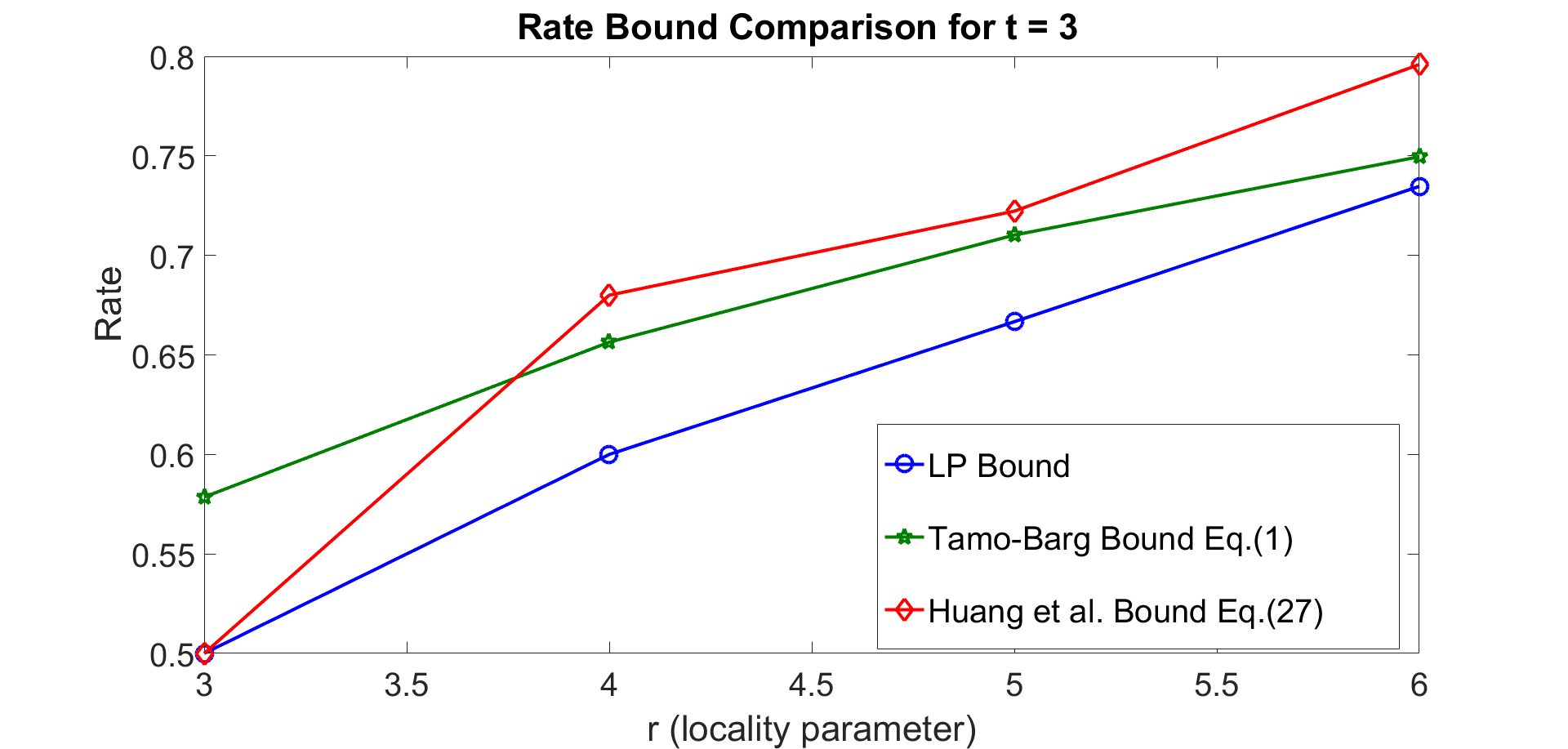}
	\caption{Comparison of the LP rate bound on binary strict availability codes with existing bounds}
	\label{fig:LPboundt3}
\end{figure}
The above formulation can be extended by adding further linear constraints obtained by linear combinations for 3 rows of the parity check matrix $H$ and so on. The above linear program \eqref{LP_11},\eqref{LP_12} was evaluated using MATLAB for $q = 2$, $n = (r+1)^2$ and compared with the bounds \eqref{TamoBargRate} and \eqref{HuangSiegelRate} and an improvement was observed for the special case of strict availability codes with $r > 2$. (see figure \ref{fig:LPboundt3})
\section{Bounds on minimum distance of codes with $t$ availability}\label{SecMinDistBounds}
In this section, we present field-size dependent and field-size independent bounds on minimum distance of $(n,k,r,t)_{a}$ codes.  Let $d^q_{\text{min}}(n,k,r,t)$ denote the maximum possible minimum distance of an $(n,k,r,t)_a$ code over the field $\mathbb{F}_q$. Let $d_{\text{min}}(n,k,r,t)$ denote the maximum possible minimum distance of an $(n,k,r,t)_a$ code independent of field size. Two field-size independent bounds are available in literature:
In \cite{TamBarFro} the following bound on minimum distance of an $(n,k,r,t)_{a}$ code was presented:
\bea
d_{\text{min}}(n,k,r,t) \leq n-\sum_{i=0}^{t} \left \lfloor \frac{k-1}{r^i} \right \rfloor \label{TamoBargDmin}
\eea
In \cite{SonYue_Square_Code} the following bound on minimum distance of a code with information symbol availability was presented:
\bea
d_{\text{min}}(n,k,r,t) \leq n-k+2 - \left \lceil \frac{t(k-1)+1}{t(r-1)+1} \right \rceil \label{WangDmin}
\eea
We compare our field size independent minimum distance bounds with the above two bounds.
\vspace{-0.2cm}
\subsection{Field-Size Dependent Bound on Minimum Distance of an $(n,k,r,t)_{a}$ Code}
Here we present field-size dependent bound on minimum distance of an $(n,k,r,t)_{a}$ code. We calculate it using Generalized Hamming Weights (GHW) of the dual of an $(n,k,r,t)_{a}$ code.
\begin{thm} \label{thm:dmin1}
	Let $\mathcal{C}$ be an $(n,k,r,t)_{a}$ code over a field $\mathbb{F}_q$ with minimum distance $d^q_{\text{min}}(n,k,r,t)$ then:
	\bea
	d^q_{\text{min}}(n,k,r,t) \leq  \min_{i \in S} \ d^q_{\text{min}}(n-e_i,k+i-e_i,r,t), \ \label{Mindist_Bound1}
	\eea
	where $S=\{ i: e_i-i < k, 1 \leq i \leq b \}$ and b is any integer such that $n-k \geq b \geq 1$ and $\{e_i: 1 \leq i \leq b\}$ are a set of numbers such that that $d_i^{\perp} \leq e_i, \forall 1 \leq i \leq b$ and $\{ d_i^{\perp} : 1 \leq i \leq n-k \}$ are the GHWs of $\mathcal{C}^{\perp}$.
\end{thm}
\begin{proof}
	Let $1 \leq i \leq b$. Let $S'=\{s_1,...,s_{d_i^{\perp}} \}$ be the co-ordinates corresponding to the support of an $i$ dimensional subspace whose support has cardinality exactly $d_i^{\perp}$ in $\mathcal{C}^{\perp}$. Add $e_i-d_i^{\perp}$ arbitrary extra co-ordiantes to $S'$ and let the resulting set be $S$. Now shorten the code $\mathcal{C}$ in the co-ordinates given by $S$ i.e., take $\mathcal{C}^S=\{c|_{S^c} : c \in \mathcal{C}, c|_S = 0\}$ . $c|_{A}$ refers to the code symbols in the codeword $c$ corresponding to the co-ordinates in the set $A$ .The resulting code $\mathcal{C}^S$ has block length $n-e_i$ and dimension $\geq n-e_i-(n-k-i)=k+i-e_i$ and minimum distance $\geq d^q_{min}(n,k,r,t)$ (if $k+i-e_i > 0$) and locality $r$ and availability $t$. Hence:
	\bean
	d^q_{\text{min}}(n,k,r,t) \leq  \min_{i \in S} \ d^q_{\text{min}}(n-e_i,k+i-e_i,r,t).
	\eean 
\end{proof}
\vspace{-0.35cm}
Following is an example for calculating $e_i$.
By using the expression for $e_i$ given in the following example, we will get a tighter bound on minimum distance of a code with $t$ availability over $F_q$.
\begin{example} \label{Recur1}
	
	We can calculate $e_i$ using the recursion given in \cite{PraLalKum} (given below) for calculating upper bounds on GHWs with slight modification and using the rate bound given in \cite{TamBarFro} for $t>3$ and rate bound given in \cite{SonYue_3_Erasure} for $t=3$ and  rate bound given in \cite{PraLalKum} for $t=2$:	
	\bea
	R'(r,t)&=&\begin{cases}
		\frac{r}{r+2} & \text{ if } t=2\\
		\frac{r^2}{(r+1)^2} & \text{ if } t=3 \\
		\frac{1}{\prod_{j=1}^{t}(1+\frac{1}{jr})} & \text{ if } t>3	,				
	\end{cases} \label{Rate_Bound} \\
	b  & = & \left \lceil n(1-R'(r,t)) \right \rceil, \label{paritychecks1} \\
	e_b & = & n, \label{e1}
	\eea \vspace{-0.6cm}
	\bea
	\text{For } b \geq i \geq 2, \ \ e_{i-1} = \min(e_i, e_i-\left \lceil \frac{2e_i}{i} \right \rceil + r+1). \label{e2}
	\eea	
	The expression for $e_i$ given in \eqref{paritychecks1},\eqref{e1}, \eqref{e2} for calculating an upper bound on $d_i^{\perp}$ ($i^{th}$ GHW of $\mathcal{C}^{\perp}$) of an $(n,k,r,t)_a$ code $\mathcal{C}$ can be obtained by taking the matrix $H_{des}(\mathcal{C})$ and taking a set of $b$ (Eq. \eqref{paritychecks1}) linearly independent rows in $H_{des}(\mathcal{C})$ and forming a matrix $H'_{des}(\mathcal{C})$ with these $b$ linearly independent rows and applying Lemma 5.4 in \cite{PraLalKum} to the support sets of the rows of $H'_{des}(\mathcal{C})$. Note that the upper bound $e_i$ \eqref{paritychecks1},\eqref{e1}, \eqref{e2} continues to apply even if some rows of $H'_{des}(\mathcal{C})$ have weight $< r+1$. One can always take a set of $b$ (Eq. \eqref{paritychecks1}) linearly independent rows in $H_{des}(\mathcal{C})$ because the equation given for $b$ (Eq. \eqref{paritychecks1}) is the minimum possible number of linearly independent rows (codewords) in $H_{des}(\mathcal{C})$. When we apply Lemma 5.4 in \cite{PraLalKum} as mentioned above it only gives the recursion without the min in \eqref{e2} but one can see the min can be put as in \eqref{e2} by a simple observation in the proof of Lemma 5.4 in \cite{PraLalKum}.
\end{example}
\begin{cor} \label{Cor1}
	Let $\mathcal{C}$ be an $(n,k,r,t)_{a}$ code over a field $\mathbb{F}_q$ with minimum distance $d_{min}(n,k,r,t)$. Then field-size dependency on the bound \eqref{Mindist_Bound1} can be removed and written as:
	\bea
	& d_{\text{min}}(n,k,r,t) \leq  \min_{i \in S}  \  d_{\text{min}}(n-e_i,k+i-e_i,r,t) \notag \\
	&  \leq \min_{i \in S}  \  \ \   n-k-i+1-\sum_{j=1}^{t} \left \lfloor \frac{k+i-e_i-1}{r^j} \right \rfloor \ \   \ \  \ \label{Mindist_Bound2}
	\eea
	where $S= \{ i: e_i-i < k, 1 \leq i \leq b \}$ and $e_i,1 \leq i \leq b $ are calculated using example \ref{Recur1} using \eqref{paritychecks1},\eqref{e1},\eqref{e2}.\\
\end{cor}
	  Equation \eqref{Mindist_Bound2} is derived by substituting the bound \eqref{TamoBargDmin} for $d_{\text{min}}(n-e_i,k+i-e_i,r,t)$. Note that the calculation of $e_i$ using \eqref{paritychecks1},\eqref{e1},\eqref{e2} is independent of the field $\mathbb{F}_q$ and applies even if some rows in $H_{des}\mathcal{(C)}$ have weight $< r+1$. Since \eqref{Mindist_Bound2} is a simple observation on \eqref{Mindist_Bound1}, we skip the proof. 

   	
	\begin{note}\textbf{Tightness of the Bound}:
	We use \eqref{Mindist_Bound2} along with expression for $e_i$ given in Eg. \ref{Recur1} (Eq.\eqref{paritychecks1},\eqref{e1},\eqref{e2}) to calculate a bound on $d_{min}(n,k,r,t)$. The resulting bound is tighter than the bounds \eqref{TamoBargDmin},\eqref{WangDmin}. We plot our bound for $t=3$ in Fig.~\ref{fig:dmin_t3}. It can be seen from the plot in Fig.~\ref{fig:dmin_t3} that our bound is tigher than the bounds \eqref{TamoBargDmin},\eqref{WangDmin}.
	\end{note}
	
	\begin{figure}[h!]
		\centering
		\includegraphics[width=5in]{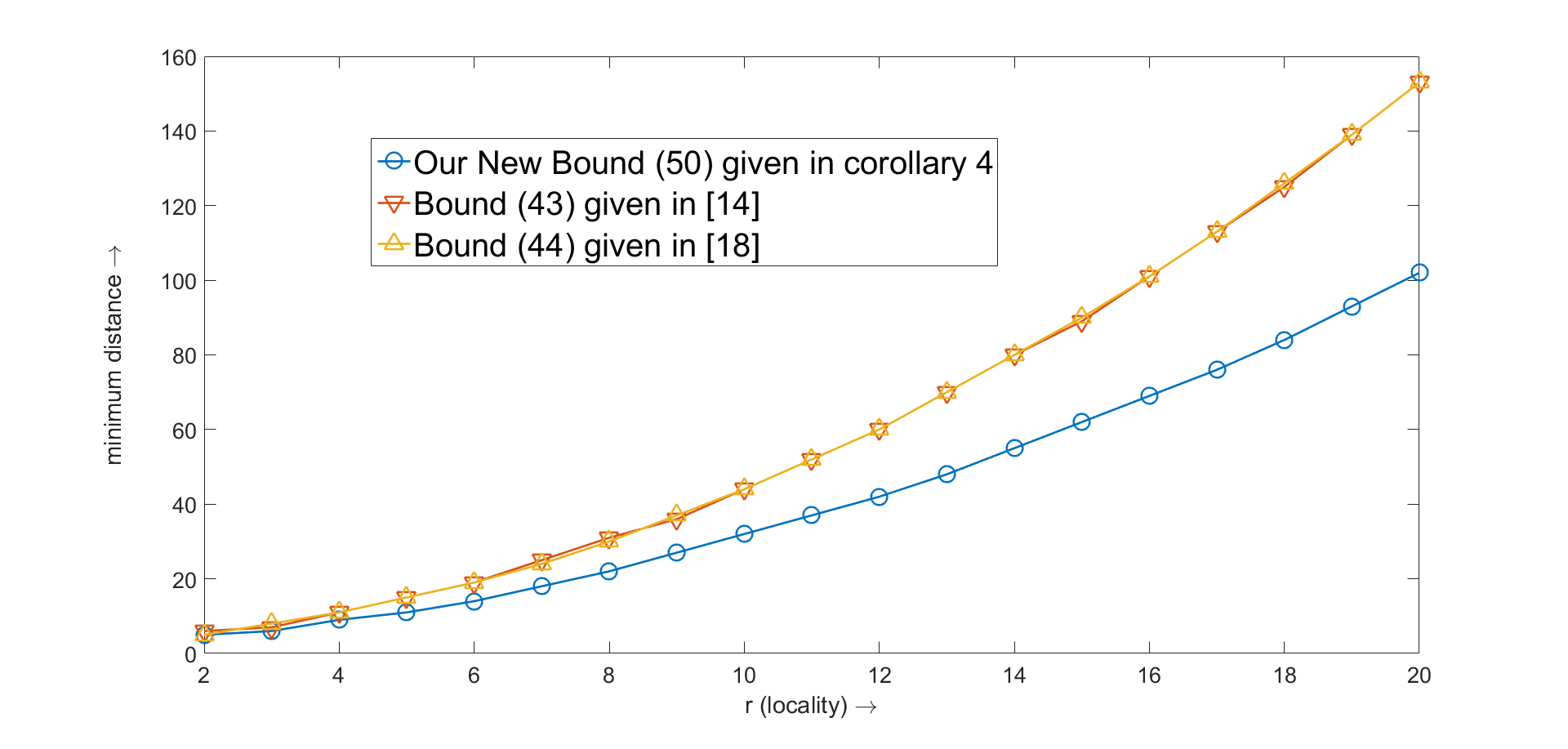}
		\caption[Example plot of locality vs minimum distance]{Plotting locality vs minimum distance for $t=3$ with $n= {r+3 \choose 3},k=\frac{nr}{r+3}$. Here we are comparing our bound \eqref{Mindist_Bound2} given in Corollary \ref{Cor1} with the bounds \eqref{TamoBargDmin},\eqref{WangDmin}. In the plot, bounds \eqref{TamoBargDmin},\eqref{WangDmin} overlap at lot of points. Note that our bound \eqref{Mindist_Bound2} is tighter than the bounds \eqref{TamoBargDmin},\eqref{WangDmin}.}
		\label{fig:dmin_t3}
	\end{figure}

\begin{cor} \label{Cor2}
	
	Let $\mathcal{C}$ be an $(n,k,r,t \geq 2)_{a}$ code over a field $\mathbb{F}_q$ and $H_{des}(\mathcal{C})$ be the corresponding local parity check matrix as defined in introduction section with only local parity checks.  
	Let $\mathcal{B}_{0} =\text{RowSpace}(H_{des}(\mathcal{C}))$.
	Let the dimension of $\mathcal{B}_{0}$ be $M$. If the code  $\mathcal{B}_{0}$ has a generator matrix with full rank such that the Hamming weight of all the columns of the matrix is $\leq \delta$ and all the rows have Hamming weight $\leq r+1$ then, 
	
	$e_i$ (an upper bound on $d^{\perp}_i$ ($i^{th}$ GHW of $\mathcal{C}^{\perp}$)), $\forall 1 \leq i \leq M$ can be calculated as follows: (Since the calculation of $e_i$ is a function of $M,\delta$, we refer to it as $e_i(M,\delta)$.)
	\bea
	e_1(M,\delta) &= & r+1,  \ \ J_1 = 0, \notag \\
	\text{For } M \geq i & \geq & 2 :  \notag \\
	J'_{1i} &=& r+1 - \left \lfloor \frac{\delta(n-e_{i-1}(M,\delta))}{M-i+1} \right \rfloor, \notag \\
	J'_{2i} &=& \left \lceil \frac{2e_{i-1}(M,\delta)-(i-1)-(i-1)(r+1)}{M-i+1} \right \rceil, \notag
	\eea	 
	\begin{align} 
	J_i = \begin{cases}
	\text{max}(J'_{1i},J'_{2i},1) & \text{if } F(M,\delta) \geq M, r+1-J_{i-1} \geq 2 \\
	\text{max}(J'_{1i},J'_{2i},0) & \text{if } F(M,\delta) \geq M, r+1-J_{i-1} < 2 \\
	\text{max}(J'_{1i},1) & \text{if } F(M,\delta) < M, r+1-J_{i-1} \geq 2 \\
	\text{max}(J'_{1i},0) & \text{if } F(M,\delta) < M, r+1-J_{i-1} < 2,
	\end{cases} \notag
	\end{align}
	\hspace{5cm}
	where $F(M,\delta)= n-e_{i-1}(M,\delta)$,
	\bea
	e_i(M,\delta) &=& e_{i-1}(M,\delta) + r+1 - J_i \label{Recurr_eM}
	\eea
	and Let minimum distance of $\mathcal{C}$ be $d_{\text{min}}(M,\delta)(n,k,r,t)$ then using \eqref{Mindist_Bound1},\eqref{TamoBargDmin} as in Corollary \ref{Cor1}: 
	\bea
	\hspace{-5cm}
	d_{\text{min}}(M,\delta)(n,k,r,t) & \leq & \min_{i \in S}  \  \ \  d_{\text{min}}(n-e_i(M,\delta),k+i-e_i(M,\delta),r,t) \notag \\
	& \leq & \min_{i \in S}  \  \ \   n-k-i+1-\sum_{j=1}^{t} \left \lfloor \frac{k+i-e_i(M,\delta)-1}{r^j} \right \rfloor \label{Mindist_Bound3}
	\eea
	where $S=\{ i: e_i(M,\delta)-i < k, 1 \leq i \leq M \}$.
	
	For a bound on $d_{\text{min}}$ independent of $M,\delta$ we take:
	\bea
	d_{\text{min}}(n,k,r,t) & \leq & \max_{(M,\delta) \in S_1} \ \ d_{\text{min}}(M,\delta)(n,k,r,t), \notag \\
		& \leq & \max_{(M,\delta) \in S_1} \min_{i \in S}  \  \ \   n-k-i+1-\sum_{j=1}^{t} \left \lfloor \frac{k+i-e_i(M,\delta)-1}{r^j} \right \rfloor \label{Mindist_Bound4}
	\eea
	where $S_1 = \{ \left \lceil n(1-R'(r,t)) \right \rceil \leq M \leq n-k, n-k \geq \delta \geq 0 \}$ and  $S=\{ i: e_i(M,\delta)-i < k, 1 \leq i \leq M \}$.
\end{cor}   	
      \begin{proof}
      	The only thing we have to prove is that the expression \eqref{Recurr_eM} given for $e_i(M,\delta)$ gives an upper bound on $d^{\perp}_i$.\\
      Proof of equation \eqref{Recurr_eM} for calculating $e_i(M,\delta)$: \\
          To avoid cumbersome notation we refer to $e_i(M,\delta)$ as $e_i$ in the proof. In this proof we denote support of a vector v by $Supp(v)$.
          Let $\{h_1,...,h_M\}$ denote a set of $M$ codewords which are basis of $\mathcal{B}_{0}$ such that $|Supp(h_i)| \leq r+1, \forall 1 \leq i \leq M$ and the following matrix $H_1$ has all column wights $\leq \delta $ (by the assumption in the theorem):
            \bean
            H_1 & = &  \left[ 
            \begin{array}{c}
            	h_1 \\
            	\vdots \\ 
            	h_M
            \end{array}
            \right]
            \eean
           
         Lets take the codeword $h_1$ and set $e_1=r+1$. Let $A_1=Supp(h_1)$. Now assume that we have a set of co-ordinates $A_{i-1}$ such that $|A_{i-1}|=e_{i-1}$ and $\cup_{j=1}^{i-1} Supp(h_{s_j}) \subseteq A_{i-1}$ for some distinct $s_1,...,s_{i-1}$. Now we are going to select $e_i-e_{i-1}$ co-ordinates from $[n]-A_{i-1}$ and add these co-ordinates to $A_{i-1}$ to form $A_i$ such that $A_i$ contains the support of at least $i$ distinct codewords from $\{h_1,...,h_M\}$. Now write the matrix $H_1$ after permuting rows and columns such that the co-ordinates represented by $A_{i-1}$ are the first $e_{i-1}$ columns of the matrix and the first $i-1$ rows of the matrix are the rows $h_{s_1},..,h_{s_{i-1}}$ (upto permutation of columns). Hence the resulting matrix can be written as:
         \bean
         H'_1 & = &  \left[ 
         \begin{array}{c|c}
         	H_4  & 0 \\
         	H_2 & H_3		
         \end{array}
         \right] 
         \eean
         The matrix $H_4$ is the $i-1 \times e_{i-1}$ matrix with columns corresponding to the coordinates in $A_{i-1}$ and $h_{s_1},..,h_{s_{i-1}}$ as its rows (upto permutation of columns and after throwing away some zero weight columns). Now take the matrix $H_2$. Lets write the matrix $H_2$ as $H_2=[A | B]$ after permuting its columns such that  the matrix $A$ has exactly the columns corresponding to the co-ordinates $F_{i-1}=\cup_{j=1}^{i-1} Supp(h_{s_j})$. In $H'_1$, since the sum of weight of the columns indexed by $F_{i-1}=\cup_{j=1}^{i-1} Supp(h_{s_j})$ (i.e., sum of weight of columns corresponding to the co-ordinates in $F_{i-1}=\cup_{j=1}^{i-1} Supp(h_{s_j})$) is $\geq 2|F_{i-1}|-(i-1)$ (since $ t\geq 2 $, number of weight one columns  in $H'_1$ from among the columns indexed by $F_{i-1}$ is atmost $i-1$), the sum of weight of columns of $A$ is $\geq 2|F_{i-1}|-(i-1)-(i-1)(r+1)$. Now the columns of the matrix $H_4$ corresponding to the columns in $B$ has 0 weight. Hence sum weight of the columns in $B$ is at least $2(e_{i-1}-|F_{i-1}|)$. If not, we can remove these extra $e_{i-1}-|F_{i-1}|$ co-ordinates corresponding to the columns in $B$ from $A_{i-1}$ and add to $A_{i-1}$ arbitrarily chosen $e_{i-1}-|F_{i-1}|$ co-ordinates from the rest of $n-|F_{i-1}|$ co-ordinates (i.e., from co-ordinates outside $F_{i-1}=\cup_{j=1}^{i-1} Supp(h_{s_j})$) such that the columns in $H'_1$ corresponding to the chosen $e_{i-1}-|F_{i-1}|$ co-ordinates have hamming weight at least 2 in each of the column. Such a choice of the columns is always possible at each step $i-1$ as long as  $n-|F_{i-1}| \geq M+e_{i-1}-|F_{i-1}|$ as the number of columns having hamming weight one in $H'_1$ is atmost $M$. Hence under the condition $n-|F_{i-1}| \geq M+e_{i-1}-|F_{i-1}|$, the sum of weights of all columns in $H_2$ is $\geq 2|F_{i-1}|-(i-1)-(i-1)(r+1)+2(e_{i-1}-|F_{i-1}|) = 2e_{i-1}-(i-1)-(i-1)(r+1)$. The condition boils down to $n-e_{i-1} \geq M$. Now under this condition the average row weight of the matrix $H_2$ is at least $J''_{2i}=\frac{2e_{i-1}-(i-1)-(i-1)(r+1)}{M-i+1}$. Hence there will be a row $h'_2$ in $H_2$ with weight  $\geq \left \lceil J''_{2i} \right \rceil  =J'_{2i} $. The corresponding row in $H'_1$ will have weight $\leq r+1 - J'_{2i}$ outside the cordinates in $A_{i-1}$. 
         Now the sum of column weights in $H_3$ is atmost $\delta (n-e_{i-1})$. Hence the average weight of rows in $H_3$ is atmost $J''_{1i} = \frac{\delta (n-e_{i-1})}{M-i+1}$. Hence there is row $h'_1$ in $H_3$ with weight $\leq \left \lfloor J''_{1i} \right \rfloor =J'_{1i}$.
         If $r+1-J_{i-1} \geq 2$, the number of new co-ordinates that are added to $A_{i-2}$ to form $A_{i-1}$
          is at least 2. Hence the sum weight of all columns $H_2$ is alteast 1. Hence there is a row $h'_3$ in $H_3$ with weight $\leq r$.
          
          Now we pick a row in $H'_1$ from $[H_2 | H_3]$ part of the matrix from among the rows in $H'_1$ corresponding to $h'_1,h'_2,h'_3$ such that the row has least weight in the co-ordinates $[n]-A_{i-1}$. Adding the support of the picked row to $A_{i-1}$ we form $A'_i$ such that it contains the support of at least $i$ distinct codewords in $\{h_1,...h_M\}$ and has cardinality $\leq e_i$. Now add co-ordinates arbitrarily to $A'_i$ to form $A_i$ of cardinality exactly $e_i$. 
          
          Since $A_i$ contains the support of $i$ linearly independent codewords in the dual, it is clear the $d_i^{\perp} \leq |A_i|=e_i$ (The column and row permutations that we are doing, does not affect the upper bound calculation).
          \end{proof}
 Note that column weight $\leq \delta$ constraint in corollary \ref{Cor2} applies to some full rank generator matrix of RowSpace($H_{des}(\mathcal{C})$) and  may not apply to columns of $H_{des}(\mathcal{C})$ directly. Note that calculation of $e_i(M,\delta)$ is independent of field size $\mathbb{F}_q$.
\begin{note}\textbf{Tightness of the bound:}
           In the Fig.~\ref{fig:dmin_t31} and Fig.~\ref{fig:dmin_t32}, we plot the bound given by \eqref{Mindist_Bound3} and \eqref{Mindist_Bound4} respectively for $t=3$. It can be seen from Fig.~\ref{fig:dmin_t31} that the new bound given in \eqref{Mindist_Bound3} is tighter when plotted for a specific value of $M,\delta$ as it optimizes the bound on minimum distance for the given value of $M,\delta$. It can be seen from Fig.~\ref{fig:dmin_t32} that the bound given in \eqref{Mindist_Bound4} is still tighter than the bounds\eqref{TamoBargDmin}, \eqref{WangDmin} even after maximizing over all possible $M,\delta$. Our bound depending on $M,\delta$ might throw insight as to what $M,\delta$ is optimal for minimum distance for a given $n,k,r,t$. It may be expected that the least possible value of $M$ for a given $n,k,r,t$ to be optimal for minimum distance but we are far away from a proof of such statement. \\
          \end{note}
          
      	\begin{figure}[h!]
      		\centering

      		\includegraphics[width=5in]{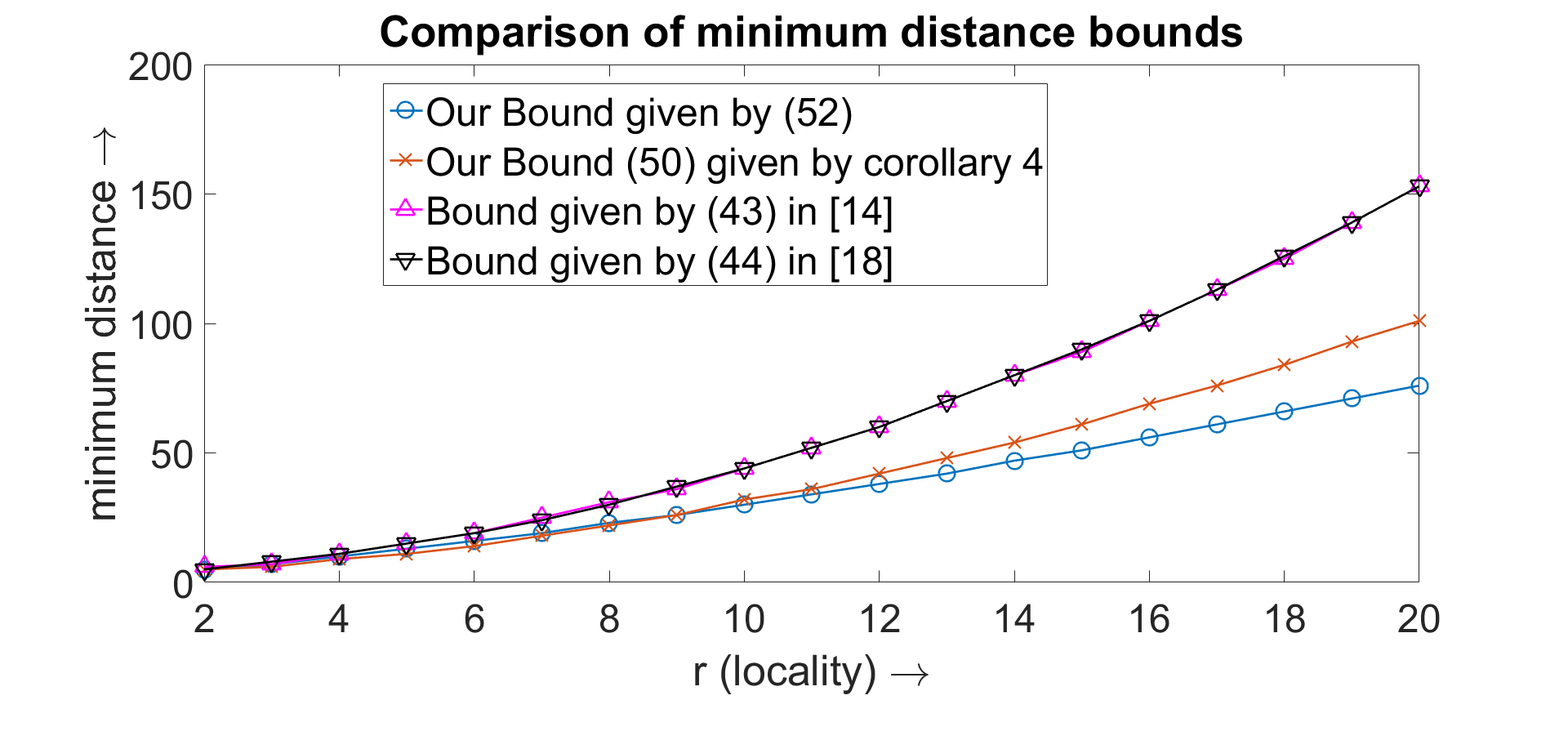}
      		\caption[Example plot of locality vs minimum distance]{Plotting locality vs minimum distance for $t=3$ with $n= {r+3 \choose 3},k=\frac{nr}{r+3}, M=n-k, \delta=t$. Here we are comparing the bounds in \eqref{Mindist_Bound3},\eqref{Mindist_Bound2} with the bounds \eqref{TamoBargDmin},\eqref{WangDmin}. Note that the bound \eqref{Mindist_Bound2} plotted above is independent of choice of $M,\delta$. In the plot, bounds \eqref{TamoBargDmin},\eqref{WangDmin} overlap at lot of points. Here we are plotting the bound in \eqref{Mindist_Bound3} with $M=n-k,\delta=t$ and this value of $M,\delta$ corresponds to the correct value for the code construction given in \cite{WanZhaLiu_Arb_Locality} for the given $n,r,t=3$.  Note that our bounds \eqref{Mindist_Bound3}, \eqref{Mindist_Bound2} are tighter than the bounds \eqref{TamoBargDmin},\eqref{WangDmin}.}
      		\label{fig:dmin_t31}
      	\end{figure}
%
%
      		      		\begin{figure}[h!]
      		      			\centering
      		      			\includegraphics[width=5in]{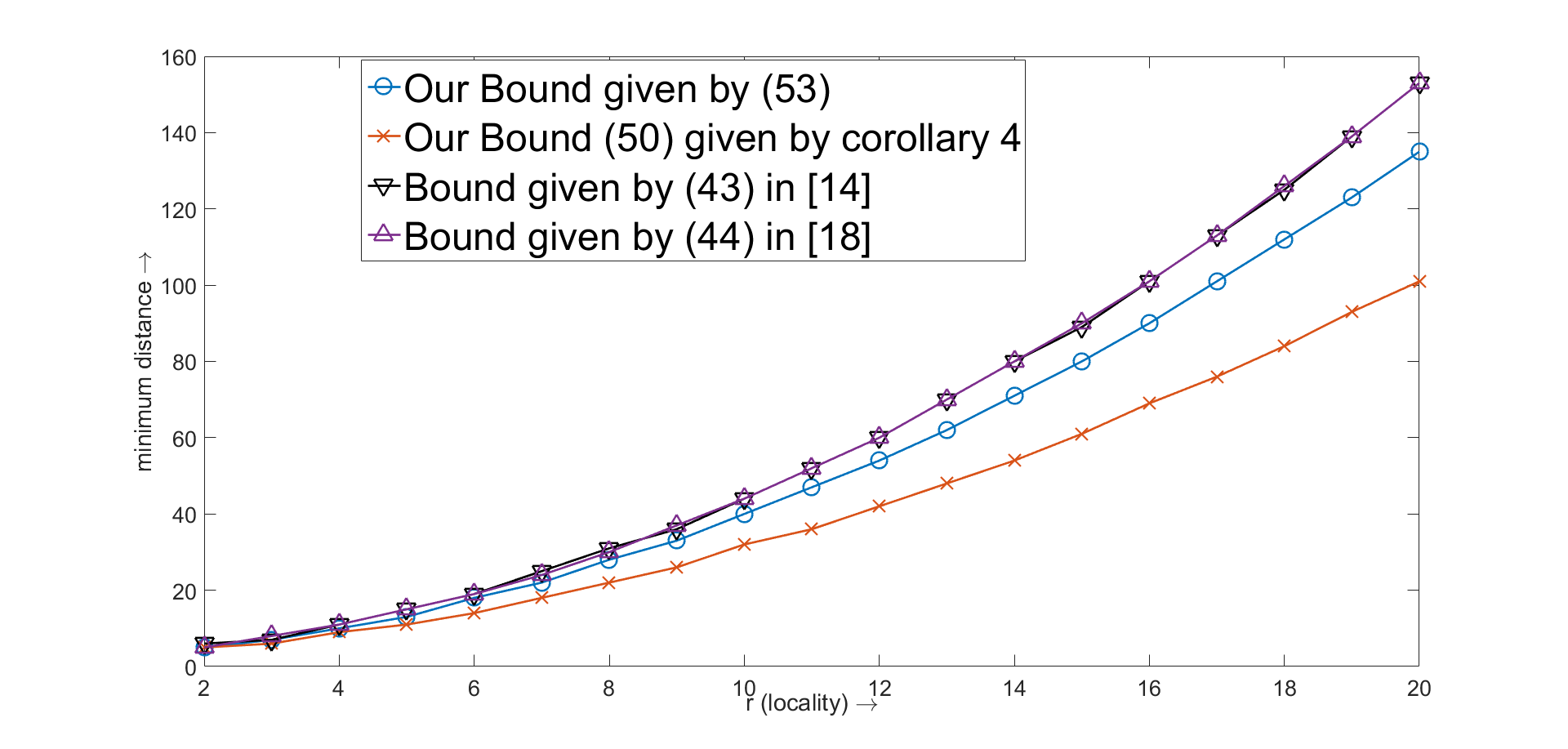}
      		      			\caption[Example plot of locality vs minimum distance]{Plotting locality vs minimum distance for $t=3$ with $n= {r+3 \choose 3},k=\frac{nr}{r+3}$.  Here we are comparing the bounds in \eqref{Mindist_Bound4},\eqref{Mindist_Bound2} with the bounds \eqref{TamoBargDmin},\eqref{WangDmin}. In the plot, bounds \eqref{TamoBargDmin},\eqref{WangDmin} overlap at lot of points.  Note that our bounds \eqref{Mindist_Bound4}, \eqref{Mindist_Bound2} are tighter than the bounds \eqref{TamoBargDmin},\eqref{WangDmin}.}
      		      			\label{fig:dmin_t32}
      		      		\end{figure}
      		
\begin{example}
	
	It can be seen from the corollary \ref{Cor2} and its proof that $d_i^{\perp} \leq ir+1$ for $r \geq 2, t \geq 2$. Hence using this upper bound, the bound in  \eqref{Mindist_Bound1} can be wriiten as:
	\bean
	d^q_{min}(n,k,r,t) \leq min_{\{ i: i(r-1)+1 < k \}}  \  \ \  d^q_{min}(n-ir-1,k-(i(r-1)+1),r,t) \\ 
	\eean
	which is tighter than the bound given in \cite{HuaYaaUchSie}. The bound given in \cite{HuaYaaUchSie} is for information symbol availability. But here we consider all symbol availability. Hence the tightening may be natural. But we would like to draw the attention to the fact that our bound can be further tightened by using more accurate upper bound $e_i$ (equation \eqref{paritychecks1},\eqref{e1},\eqref{e2}) on GHWs of dual.
	
\end{example}

\subsection{Binary $(n,k,r,t)_{sa}$ Codes Based on Partitions}
\begin{const} \label{C1}
	Let $r+1$ be a given integer. Let $\{P_1,...,P_h\}$ be a collection of partitions of $[n]$ ($n=(r+1)^g$ for some $g \geq 1$) with $P_i=\{Q_1^i,..,Q_{\frac{n}{r+1}}^i\}$, $\forall 1 \leq i \leq h$ such that $|Q_{x}^i|=r+1$, $|Q_{x}^i \cap Q_{y}^i|=0$, $\forall 1 \leq x \neq y \leq \frac{n}{r+1}$, $\forall 1 \leq i \leq h$ and $|Q_{x}^i \cap Q_{y}^j| \leq 1$, $\forall 1 \leq x,y \leq \frac{n}{r+1}$, $\forall 1 \leq i \neq j \leq h$. Let the maximum number of such partitions that can be formed be $T(n)$. i.e., $h \leq T(n)$. Then we can show by explicit recursive construction of partitions that (explained below):
	\bean
	T(n) \geq f(r+1) \ T\left(\frac{n}{r+1}\right) + 1 \Rightarrow T(n) \geq \frac{f(r+1)^g-1}{f(r+1)-1}, 
	\eean
	where $f(r+1)=N(r+1)+1$ and $N(r+1)$ is the maximum number of mutually orthogonal Latin  squares of order $r+1$.  For $1 \leq i \leq t, 1 \leq a \leq \frac{n}{r+1}$, the set $Q^i_a$, constructed using our recursive construction defines the support of a codeword $\underline{c}^i_a$ (a parity check) in the dual of our code. Form a binary matrix $H$ with the codewords $\{\underline{c}^i_a : 1 \leq i \leq t, 1 \leq a \leq \frac{n}{r+1}\}$ as rows. $H$ defines the parity check matrix of our construction. Note that for this construction to work we need to pick $t$ partitions out of $\frac{f(r+1)^g-1}{f(r+1)-1}$ partitions we constructed. Hence this construction gives an $(n,k,r,t)_{sa}$ code for some $k$, for any $t \leq \frac{f(r+1)^g-1}{f(r+1)-1}$.
\end{const}
We demonstrate the recursive construction for $r=1$. Although $r=1$ may look like a trivial case, the basic idea is clearly visible in this case and the construction for general $r$ is similar.
\ben
 \item Aim: Form a collection of partitions of $[n]=[2^g]$, $\{P_1,....P_{2^g-1}\}$ such that each set in a partition contains just 2 elements and if $\{a,b\}$ appears as a set in partition $P_i$ then the set must not appear in any other partition $P_j$, $i \neq j$.
 \item It can be recursively shown that such a collection of partition can be formed. $T(n)$ denotes the maximum number of such collection of partitions of $n=2^g$ elements. It will be seen that $T(n) \geq 2T(n/2)+1$. From this by recursively solving, we get $T(n) \geq 2^g -1$.
 The recursive bound $T(n) \geq 2T(n/2)+1$ and also the construction can be seen as follows:\\
 \ben
 \item Pick a partition $P_1$= $\{ A_1=\{a_1,a_2\},...,A_{2^{g-1}}=\{a_{2^g-1},a_{2^g}\} \}$ of $[n]$ arbitrarily.\\
 \item Let $W_{T(n/2)}= \{Q_1,...,Q_{T(n/2)}\}$ be a collection of $T(n/2)$ partitions of $[\frac{n}{2}]$ satisfying the required conditions.
 \item Now take a partition $Q_1$ from $W_{T(n/2)}$. Let the partition be: $Q_1$= $\{ B_1=\{b_1,b_2\},...,B_{2^{g-2}}=\{b_{2^{g-1}-1},b_{2^{g-1}} \} \}$.\\
  \item Now using the above partition $Q_1$ we form 2 more partitions $P_3,P_4$ of $[n]$ as follows:\\
 	Let $P_3=P_4=\emptyset$\\
 	For each set $B_j = \{b_{2j-1},b_{2j}\}$ from $Q_1$, do the following (we simply pick the sets $A_{j_1},A_{j_2}$ with $j_1=b_{2j-1}$ and $j_2=b_{2j}$  and combine the elements in $A_{j_1}=\{a_{2j_1-1},a_{2j_1}\}$ and $A_{j_2}=\{a_{2j_2-1},a_{2j_2}\}$ to form the sets as shown below):\\
 	\bean
 	   P_3 = P_3 \cup \{\{a_{2j_1-1},a_{2j_2-1}\},\{a_{2j_1},a_{2j_2}\}\} \\
 	   P_4 = P_4 \cup \{\{a_{2j_1-1},a_{2j_2}\},\{a_{2j_1},a_{2j_2-1}\}\}
 	\eean
 	The resulting collection of sets $P_3,P_4$ can be easily seen to be partitions of $[n]$.
 	Similarly we can form 2 more partitions of $[n]$ for each of the partition in $W_{T(n/2)}$. It is easy to see that the resulting collection of partitions of $[n]$ satisfy the required conditions since $W_{T(n/2)}$ satisfies the required conditions. Hence $T(n) \geq 2T(n/2)+1$.
 \een
\een

We demonstrate in the following, the recursive construction for general $r$.

Let $P_1^n = \{(Q_n^1)_1,...,(Q_n^1)_{\frac{n}{r+1}}\}$ be the natural partition of $[n]$ i.e. $\{\{1,...,r+1\},\{r+1+1,...,2(r+1)\},...,\{n-(r+1)+1,...,n\}\}$. Now assume we have constructed $h=T(\frac{n}{r+1})$ partitions $\{P_1^\frac{n}{r+1},...,P_{h}^\frac{n}{r+1}\}$ of $[\frac{n}{r+1}]$ satisfying the required intersection properties. Let one of these partitions be $P_i^\frac{n}{r+1} = \{(Q^i_\frac{n}{r+1})_1,...,(Q^i_\frac{n}{r+1})_{\frac{n}{(r+1)^2}}\}$. Each set of this partition has size $r+1$. Let $S_1,...,S_{N(r+1)}$ be mutually orthogonal Latin squares of order $r+1$. For $1 \leq i \leq N(r+1)$, the latin square $S_i$ is written as $r+1 \times r+1$ matrix. Let $S_0$ be an $r+1 \times r+1$ matrix: 
	\begin{equation}
		S_0 = \left[\begin{array}{c c c c}
			1 & 1 & \hdots & 1 \\
			2 & 2 & \hdots & 2 \\
			\vdots & \vdots & \hdots & \vdots \\
			r+1 & r+1 & \hdots & r+1 \\
		\end{array}
		\right].
	\end{equation}
	Let $S_{N(r+1)+1}$ be another matrix that is mutually orthogonal with $S_0,S_1,...,S_{N(r+1)}$ given by: 
	\begin{equation}
		S_{N(r+1)+1} = \left[\begin{array}{c c c c}
			1 & 2 & \hdots & r+1 \\
			1 & 2 & \hdots & r+1 \\
			\vdots & \vdots & \hdots & \vdots \\
			1 & 2 & \hdots & r+1 \\
		\end{array}
		\right]
	\end{equation}
	We follow the steps:\\
	For $2 \leq j \leq T(n)$, set $P^n_{j} = \emptyset$\\
	Do for all $\{(i,p,j):1 \leq i \leq h, 1 \leq j \leq N(r+1)+1, 1 \leq p \leq \frac{n}{(r+1)^2}\}$:
	\begin{enumerate}
		\item Let $(Q^i_\frac{n}{r+1})_p = \{\sigma^i_{(1,p)},...,\sigma^i_{(r+1,p)}\}$ and for $1 \leq x \leq \frac{n}{r+1}$, if $(Q^1_n)_x = \{\beta^1_{(1,x)},...,\beta^1_{(r+1,x)} \}$ then let $(\underline{q}^1_n)_x = [\beta^1_{(1,x)},...,\beta^1_{(r+1,x)}]$.
		Then form an $r+1 \times r+1$ matrix: \\
		$U = \left[\begin{array}{c}
		(\underline{q}^1_n)_{\sigma^i_{(1,p)}}\\
		\vdots\\
		(\underline{q}^1_n)_{\sigma^i_{(r+1,p)}}
		\end{array}
		\right]$.
		\item  Let $W(a,b)$ be $(a,b)^{th}$ entry of a matrix $W$. For $1 \leq x \leq r+1$, let $A^j_x = \{(a,b): S_j(a,b)=x \}$ and let $(Q^{(i-1)f(r+1)+j}_n)_{(p-1)(r+1)+x}=\{U(a,b) : (a,b) \in A^j_x\}$.\\
		For $1 \leq x \leq r+1$:\\
	    \bean
		P^n_{(i-1)f(r+1)+j} = P^n_{(i-1)f(r+1)+j} \cup \{(Q^{(i-1)f(r+1)+j}_n)_{(p-1)(r+1)+x}\}
		\eean
		i.e., $\{U(a,b) : (a,b) \in A^j_x\}$ is made as one of the sets of the partition $P^n_{(i-1)f(r+1)+j}$.
	\end{enumerate}
	$P^n_a$, $\forall 2 \leq a \leq hf(r+1)$ constructed as described above along with $P^n_1$ are the final partitions that we want.
	Hence we have explicitly constructed $hf(r+1)=T(\frac{n}{r+1}) f(r+1)$ partitions of $[n]$. It can be seen that these partitions satisfy the required intersection properties by using the fact that partitions of $[\frac{n}{r+1}]$ also satisfy the required intersection properties and the fact that we using orthogonal latin squares.
	\begin{enumerate}
	\item For all $1 \leq x \neq x' \leq r+1$: 
	$|(Q^{(i-1)f(r+1)+j}_n)_{(p-1)(r+1)+x} \cap (Q^{(i-1)f(r+1)+j}_n)_{(p-1)(r+1)+x'}| = 0$ due to the fact that $A_x^j \cap A_{x'}^j = \emptyset$.
	\item For all $1 \leq j \neq j' \leq f(r+1)$, $1 \leq x,x' \leq r+1$: 
	$|(Q^{(i-1)f(r+1)+j}_n)_{(p-1)(r+1)+x} \cap (Q^{(i-1)f(r+1)+j'}_n)_{(p-1)(r+1)+x'}| \leq 1$ due to the fact that $S_j$ and $S_{j'}$ are mutually orthogonal latin squares.
	\item For all $1 \leq p \neq p' \leq \frac{n}{(r+1)^2}$ and $1 \leq j,j' \leq f(r+1)$ and $1 \leq x,x' \leq r+1$: \\
	$|(Q^{(i-1)f(r+1)+j}_n)_{(p-1)(r+1)+x} \cap (Q^{(i-1)f(r+1)+j'}_n)_{(p'-1)(r+1)+x'}| = 0$ due to the fact that $(Q^i_\frac{n}{r+1})_p \cap (Q^i_\frac{n}{r+1})_{p'} = \emptyset$.
		\item For all $1 \leq i \neq i' \leq h$, $1 \leq p, p' \leq \frac{n}{(r+1)^2}$ and $1 \leq j,j' \leq f(r+1)$ and $1 \leq x,x' \leq r+1$: \\
		$|(Q^{(i-1)f(r+1)+j}_n)_{(p-1)(r+1)+x} \cap (Q^{(i'-1)f(r+1)+j'}_n)_{(p'-1)(r+1)+x'}| \leq 1$ due to the fact that $|(Q^{i}_\frac{n}{r+1})_{p} \cap (Q^{i'}_\frac{n}{r+1})_{p'}| \leq 1$ and the fact that we pick only one element from each row of $U$ by using latin squares to form a set of a partition.
	\end{enumerate}
	 Hence we have $T(n) \geq f(r+1) \ T\left(\frac{n}{r+1}\right)$ at the same time explicitly constructed $f(r+1) \ T\left(\frac{n}{r+1}\right)$ partitions satisfying the required conditions recursively.
%

For $r+1=q=p^m$, for $p$ a prime amd $m \geq 1$, the parameters $n,r,t$ of our construction match with that of a construction based on incidence matrix of lines and points in $F_{q^g}$. Our construction for $r+1=q=p^m$ has rate higher than the construction given in \cite{WanZhaLiu_Arb_Locality}. For picking $t$ partitions, we are left with the choice of picking it from $\frac{f(r+1)^g-1}{f(r+1)-1}$ partitions. We are currently exploring on the optimal choice of partitions.
So we do not make any explicit rate comparisons here. We leave this as open problem for future work. Our construction can also be used to form a matrix $H_1=[I \ | H]$ where $H_1$ defines a parity check matrix of an information symbol availability code which can be used in the setting \cite{RawPapDimVis_arxiv}.

\vspace{-0.2cm}
\subsection{Binary $(n,k,r,t)_{sa}$ Codes Based on Functions with Certain Structure}
\begin{const} \label{C2}
	Let $X = \{x_1,...,x_n\}$, $Y = \{y_1,...,y_l\}$. Let $\mathcal{F} = \{f_1,...,f_t\}$ where $f_i:X \rightarrow Y$, for $1 \leq i \leq t$. Define inverse set of a function as $f_i^{-1}(y) = \{x:f_i(x)=y\}$. Choose $\mathcal{F}$ such that $|f_i^{-1}(y)| \leq r+1$ for every $f_i \in \mathcal{F}$, $y \in Y$ and $|f_i^{-1}(y_l) \cap f_j^{-1}(y_p)| \leq 1$ for every $f_i \neq f_j \in \mathcal{F}$, or $f_i = f_j$ and $y_l \neq y_p \in Y$. We will define a code based on the above quantities by describing its parity check matrix as:
	\begin{itemize}
		\item Let $H$ be an $lt \times n$ matrix over $\mathbb{F}_2$. Index the columns of $H$ by distinct elements of $X$. Index the rows by 2-tuples $(f_i,y_j)$, $1 \leq i \leq t$ and $1 \leq j \leq l$.
		\item The entry $H((f_i,y),x)$ is 1 if $f_i(x)=y$ and 0 otherwise for every $x \in X, f_i \in \mathcal{F}, y \in Y$.
	\end{itemize}
	From the definition of $H$, it can be seen that $H$ is a parity check matrix of an $(n,k,r,t)_{sa}$ code over $\mathbb{F}_2$ for some $k$. 
\end{const}
\begin{example} \label{Ex22}
	For $q$ a power of a prime, Let $X = \mathbb{F}_q^{n_1}, Y = \mathbb{F}_q^{m_1}$ with $2m_1 \geq n_1, m_1 < n_1$. For $\underline{x} \in X$, the functions we choose are $f_i(\underline{x}) = A_i\underline{x}$, $1 \leq i \leq t$ where $A_i$ is an $m_1 \times n_1$ matrix over $\mathbb{F}_q$ such that $rank([\frac{A_i}{A_j}])=n_1$ for all $t \geq i \neq j \geq 1$ and $rank(A_i)=m_1$, $\forall 1 \leq i \leq t$ . It can be seen that the matrix $H$ formed as described in Construction \ref{C2} using $\{f_1=A_1,...,f_t=A_t\},X=\mathbb{F}_q^{n_1},Y=\mathbb{F}_q^{m_1}$ describes a strict $t$ availability code with $n=q^{n_1},r+1=q^{n_1-m_1}$. A special case of our construction for $m_1 = 1, n_1 = 2$ yields a construction with rate higher than the construction in \cite{WanZhaLiu_Arb_Locality} and it turns out it is based on orthogonal Latin squares \cite{OLatin}. We are currently exploring on which set of $A_i$, $\forall 1 \leq i \leq t$ are optimal. We therefore do not make any explicit rate comparisons. We can also take $[I | H]$ to be a parity check matrix of an information symbol availability code which can be used in the setting \cite{RawPapDimVis_arxiv}.
\end{example}

\bibliographystyle{IEEEtran}
\bibliography{bib_file}	

\begin{thebibliography}{10}
\providecommand{\url}[1]{#1}
\csname url@samestyle\endcsname
\providecommand{\newblock}{\relax}
\providecommand{\bibinfo}[2]{#2}
\providecommand{\BIBentrySTDinterwordspacing}{\spaceskip=0pt\relax}
\providecommand{\BIBentryALTinterwordstretchfactor}{4}
\providecommand{\BIBentryALTinterwordspacing}{\spaceskip=\fontdimen2\font plus
\BIBentryALTinterwordstretchfactor\fontdimen3\font minus
  \fontdimen4\font\relax}
\providecommand{\BIBforeignlanguage}[2]{{%
\expandafter\ifx\csname l@#1\endcsname\relax
\typeout{** WARNING: IEEEtran.bst: No hyphenation pattern has been}%
\typeout{** loaded for the language `#1'. Using the pattern for}%
\typeout{** the default language instead.}%
\else
\language=\csname l@#1\endcsname
\fi
#2}}
\providecommand{\BIBdecl}{\relax}
\BIBdecl

\bibitem{WanZhaLiu_Arb_Locality}
A.~Wang, Z.~Zhang, and M.~Liu, ``Achieving arbitrary locality and availability
  in binary codes,'' in \emph{Proc. IEEE Int. Symp. Inf. Theory (ISIT)}, 2015,
  pp. 1866--1870.

\bibitem{CalderBankKadhe}
S.~Kadhe and R.~Calderbank, ``Rate optimal binary linear locally repairable
  codes with small availability,'' \emph{https://arxiv.org/abs/1701.02456},
  2017.

\bibitem{GopHuaSimYek}
P.~Gopalan, C.~Huang, H.~Simitci, and S.~Yekhanin, ``{On the Locality of
  Codeword Symbols},'' \emph{IEEE Trans. Inf. Theory}, vol.~58, no.~11, pp.
  6925--6934, Nov. 2012.

\bibitem{PapDim}
D.~Papailiopoulos and A.~Dimakis, ``Locally repairable codes,'' in \emph{Proc.
  IEEE Int. Symp. Inf. Theory (ISIT)}, July 2012, pp. 2771--2775.

\bibitem{OggDat}
F.~Oggier and A.~Datta, ``Self-repairing homomorphic codes for distributed
  storage systems,'' in \emph{INFOCOM, 2011 Proceedings IEEE}, April 2011, pp.
  1215--1223.

\bibitem{HuaChenLi}
C.~Huang, M.~Chen, and J.~Li, ``Pyramid codes: Flexible schemes to trade space
  for access efficiency in reliable data storage systems,'' in \emph{Network
  Computing and Applications, 2007. NCA 2007. Sixth IEEE International
  Symposium on}, July 2007, pp. 79--86.

\bibitem{KamPraLalKum}
G.~Kamath, N.~Prakash, V.~Lalitha, and P.~Kumar, ``Codes with local
  regeneration,'' in \emph{Information Theory and Applications Workshop (ITA),
  2013}, Feb 2013, pp. 1--5.

\bibitem{TamBar_Optimal_LRC}
I.~Tamo and A.~Barg, ``A family of optimal locally recoverable codes,''
  \emph{IEEE Trans. Inf. Theory}, vol.~60, no.~8, pp. 4661--4676, 2014.

\bibitem{PraLalKum}
\BIBentryALTinterwordspacing
N.~Prakash, V.~Lalitha, and P.~V. Kumar, ``Codes with locality for two
  erasures,'' 2014. [Online]. Available: \url{http://arxiv.org/abs/1401.2422}
\BIBentrySTDinterwordspacing

\bibitem{SonDauYueLi}
W.~Song, S.~H. Dau, C.~Yuen, and T.~Li, ``Optimal locally repairable linear
  codes,'' \emph{Selected Areas in Communications, IEEE Journal on}, vol.~32,
  no.~5, pp. 1019--1036, May 2014.

\bibitem{WangZhanLin}
A.~Wang, Z.~Zhang, and D.~Lin, ``Two classes of (r, t)-locally repairable
  codes,'' in \emph{Proc. IEEE Int. Symp. Inf. Theory (ISIT)}, July 2016, pp.
  445--449.

\bibitem{ZhaWanGe}
\BIBentryALTinterwordspacing
J.~Zhang, X.~Wang, and G.~Ge, ``Some improvements on locally repairable
  codes,'' 2015. [Online]. Available: \url{http://arxiv.org/abs/1506.04822}
\BIBentrySTDinterwordspacing

\bibitem{HuaYaaUchSie}
\BIBentryALTinterwordspacing
P.~Huang, E.~Yaakobi, H.~Uchikawa, and P.~H. Siegel, ``Binary linear locally
  repairable codes,'' 2015. [Online]. Available:
  \url{http://arxiv.org/abs/1511.06960}
\BIBentrySTDinterwordspacing

\bibitem{TamBarFro}
\BIBentryALTinterwordspacing
I.~Tamo, A.~Barg, and A.~Frolov, ``Bounds on the parameters of locally
  recoverable codes,'' 2015. [Online]. Available:
  \url{http://arxiv.org/abs/1506.07196}
\BIBentrySTDinterwordspacing

\bibitem{WanZha_Combinatorial_Repair_locality}
A.~Wang and Z.~Zhang, ``Repair locality from a combinatorial perspective,'' in
  \emph{Proc. IEEE Int. Symp. Inf. Theory (ISIT)}, 2014, pp. 1972--1976.

\bibitem{JuaHolOgg}
L.~Pamies-Juarez, H.~Hollmann, and F.~Oggier, ``Locally repairable codes with
  multiple repair alternatives,'' in \emph{Proc. IEEE Int. Symp. Inf. Theory
  (ISIT)}, July 2013, pp. 892--896.

\bibitem{RawPapDimVis_arxiv}
\BIBentryALTinterwordspacing
A.~S. Rawat, D.~S. Papailiopoulos, A.~G. Dimakis, and S.~Vishwanath, ``Locality
  and availability in distributed storage,'' 2014. [Online]. Available:
  \url{http://arxiv.org/abs/1402.2011}
\BIBentrySTDinterwordspacing

\bibitem{SonYue_Square_Code}
A.~Wang and Z.~Zhang, ``Repair locality with multiple erasure tolerance,''
  \emph{IEEE Trans. Inf. Theory}, vol.~60, no.~11, pp. 6979--6987, 2014.

\bibitem{SonYue_3_Erasure}
\BIBentryALTinterwordspacing
W.~Song and C.~Yuen, ``Locally repairable codes with functional repair and
  multiple erasure tolerance,'' 2015. [Online]. Available:
  \url{http://arxiv.org/abs/1507.02796}
\BIBentrySTDinterwordspacing

\bibitem{HaoRecursive}
J.~Hao, S.~Xia, and B.~Chen, ``Recursive bounds for locally repairable codes
  with multiple repair groups,'' in \emph{{ISIT}}.\hskip 1em plus 0.5em minus
  0.4em\relax {IEEE}, 2016, pp. 645--649.

\bibitem{OLatin}
M.~Y. Hsiao, D.~C. Bossen, and R.~T. Chien, ``Orthogonal latin square codes,''
  \emph{Coding Theory and Application}, vol.~14, no.~4, p. 390, 1970.

\end{thebibliography}

\end{document}